%% file: ska.tex
\documentclass[onecolumn,nofootinbib,thightenlines,notitlepage,tightenlines,longbibliography,11pt]{revtex4-1} % the option longbibliography implies that the names of the paper are shown 
\newcommand{\pagenumbaa}{1}
\usepackage{graphicx}
\usepackage{amssymb}
\usepackage{amstext}
\usepackage{tikz}
\usetikzlibrary{chains}
\usetikzlibrary{fit}
\usepackage{pgflibraryarrows}		%optional
\usepackage{pgflibrarysnakes}		%optional
\usepackage{xcolor}
\usepackage{epsfig}
\usetikzlibrary{shapes.symbols,patterns} % for source symbols

\usepackage{amsthm}
\usepackage{amsmath}
\usepackage{amsfonts}
\usepackage{amssymb,amstext}
\usepackage{bbm} % for bbm numbers
\usepackage{dsfont}

\usepackage[colorlinks=true,urlcolor=blue, hyperindex,plainpages=false,pdfpagelabels,breaklinks] {hyperref}
\usepackage{todonotes}

\usepackage{footnote}

\theoremstyle{plain}
\newtheorem{mythm}{Theorem}
\newtheorem{myprop}[mythm]{Proposition}
\newtheorem{mycor}[mythm]{Corollary}
\newtheorem{mylem}[mythm]{Lemma}

\theoremstyle{definition}

\newtheorem{myquestion}{Question}

\newcommand{\norm}[1]{\left\lVert#1\right\rVert}

\def\setC{\mathsf{c}}

\newcommand{\Hc}[2]{H\!\left({#1}\!\left|{#2}\right.\right)} %cond Entropy
\newcommand{\I}[2]{I\!\left({#1};{#2}\right)} %mutual Information
\newcommand{\Ic}[3]{I\!\left({#1};{#2}\!\left|{#3} \right. \right)} %mutual Information

\newcommand{\KL}[2]{D\!\left(\! \left.\left.{#1}\right|\! \right|{#2}\right)} %mutual Information

\newcommand{\SK}[3]{S_{\to}\!\left({#1};{#2}\!\left|{#3} \right. \right)} %one-way secret-key rate

\newcommand{\W}{\mathsf{W}} %for channel law
\newcommand{\Prv}[1]{\,{\rm Pr}\!\left[#1\right]} %usage: Pr[X \leq 5]

\DeclareMathOperator{\enc}{enc}
\DeclareMathOperator{\dec}{dec}
\newcommand{\TD}[2]{\delta\!\left({#1},{#2}\right)} %fidelity

\newcommand{\markovDavid}{\small{\mbox{$-\hspace{-1.3mm} \circ \hspace{-1.3mm}-$}}}

\usepackage{makecell}

\begin{document}

%%%%%%%%%%%%%%%%%%%%%%%%%%%%%%%%%%%%%%%%%%%%%%%%%%%%%%%%%%%%%%%%%%%%
% Place your title here
%%%%%%%%%%%%%%%%%%%%%%%%%%%%%%%%%%%%%%%%%%%%%%%%%%%%%%%%%%%%%%%%%%%%
\title{Efficient One-Way Secret-Key Agreement and Private Channel Coding via Polarization}

%%%%%%%%%%%%%%%%%%%%%%%%%%%%%%%%%%%%%%%%%%%%%%%%%%%%%%%%%%%%%%%%%%%%
% If using RevTex
%%%%%%%%%%%%%%%%%%%%%%%%%%%%%%%%%%%%%%%%%%%%%%%%%%%%%%%%%%%%%%%%%%%%
 \author{David Sutter}
 \email[]{suttedav@phys.ethz.ch}
 %\affiliation{Institute for Theoretical Physics, ETH Zurich, Switzerland}
 
 \author{Joseph M.\ Renes}
 \email[]{renes@phys.ethz.ch}
 %\affiliation{Institute for Theoretical Physics, ETH Zurich, Switzerland}

 \author{Renato Renner}
 \email[]{renner@phys.ethz.ch}
 \affiliation{Institute for Theoretical Physics, ETH Zurich, Switzerland}

%%%%%%%%%%%%%%%%%%%%%%%%%%%%%%%%%%%%%%%%%%%%%%%%%%%%%%%%%%%%%%%%%%%%
% If using IEEEtran
%%%%%%%%%%%%%%%%%%%%%%%%%%%%%%%%%%%%%%%%%%%%%%%%%%%%%%%%%%%%%%%%%%%%
%\author{David Sutter, Joseph M.\ Renes, and Renato Renner
%
%\thanks{David Sutter, Joseph M.\ Renes, and Renato Renner are with the Institute for Theoretical Physics, ETH Zurich, Switzerland (e-mail: \{suttedav,\ renes,\ renner\}@phys.ethz.ch).}

%\thanks{This work was supported by the Swiss National Science Foundation (through the National Centre of Competence in Research `Quantum Science and Technology' and grant No.~200020-135048) and by the European Research Council (grant No.~258932).}
%}

%\maketitle

%%%%%%%%%%%%%%%%%%%%%%%%%%%%%%%%%%%%%%%%%%%%%%%%%%%%%%%%%%%%%%%%%%%%
% Abstract to the abstract
%%%%%%%%%%%%%%%%%%%%%%%%%%%%%%%%%%%%%%%%%%%%%%%%%%%%%%%%%%%%%%%%%%%%

\begin{abstract}
We introduce explicit schemes based on the polarization phenomenon for the tasks of one-way secret key agreement from common randomness and private channel coding. 
For the former task, we show how to use common randomness and insecure one-way communication to obtain a strongly secure key such that the key construction has a complexity essentially linear in the blocklength and the rate at which the key is produced is optimal, i.e.,\ equal to the one-way secret-key rate. For the latter task, we present a private channel coding scheme that achieves the secrecy capacity using the condition of strong secrecy and whose encoding and decoding complexity are again essentially linear in the blocklength. 
\end{abstract}

 \maketitle %(If using RevTEx maketitle is here)
%\begin{IEEEkeywords}
%Secret-key agreement, private channel coding, one-way secret-key rate, secrecy capacity, wiretap channel scenario, more capable, less noisy, degraded, polarization phenomenon, polar codes, practically efficient, strongly secure
%\end{IEEEkeywords}

%Only for RevTEx
 \setcounter{page}{\pagenumbaa}  
 \thispagestyle{plain}

%\vspace{15mm}

\section{Introduction}\label{sec:intro}
Consider two parties, Alice and Bob, connected by an authentic but otherwise fully insecure communication channel. It has been shown that without having access to additional resources, it is impossible for them to carry out information-theoretically secure private communication \cite{shannon49,maurer93}. In particular they are unable to generate an unconditionally secure key with which to encrypt messages transmitted over the insecure channel. However, if Alice and Bob have access to correlated randomness about which an adversary (Eve) has only partial knowledge, the situation changes completely: information-theoretically secure secret-key agreement and private communication become possible.
Alternatively, if Alice and Bob are connected by a noisy discrete memoryless channel (DMC) to which Eve has only limited access---the so-called \emph{wiretap channel scenario} of 
Wyner \cite{wyner75}, Csisz\'ar and K\"orner \cite{csiszar78}, and Maurer \cite{maurer93}---private communication is again possible. 

In this paper, we present  explicit schemes for efficient one-way secret-key agreement from common randomness and for private channel coding in the wiretap channel scenario. Our schemes are based on \emph{polar codes}, a family of capacity-achieving linear codes, introduced by Ar{\i}kan \cite{arikan09}, that can be encoded and decoded efficiently.
Previous work by us in a quantum setup~\cite{concEntagDist} already implies that \emph{practically efficient} one-way secret-key agreement and private channel coding in a classical setup is possible, where a practically efficient scheme is one whose computational complexity is essentially linear in the blocklength. The aim of this paper is to explain the schemes in detail and give a purely classical proof that the schemes are reliable, secure, practically efficient and achieve optimal rates. Section~\ref{sec:background} introduces the problems of performing \emph{one-way secret-key agreement} and \emph{private channel coding}. We summarize known and new results about the optimal rates for these two problems for different wiretap channel scenarios. In Section~\ref{sec:ska}, we explain how to obtain one-way secret-key agreement that is practically efficient, strongly secure, reliable, and achieves the one-way secret-key rate. However, we are not able to give an efficient algorithm for code construction. Section~\ref{sec:pcc} introduces a similar scheme that can be used for strongly secure private channel coding at the secrecy capacity. Finally in Section~\ref{sec:discussion}, we state two open problems that are of interest in the setup of this paper as well as in the quantum mechanical scenario introduced in~\cite{concEntagDist}.

%\prlsection{Notation} 
\section{Background and Contributions} \label{sec:background}
\subsection{Notation and Definitions}
\label{sec:defn}
Let $[k]=\left \lbrace 1,\ldots,k \right \rbrace$ for $k\in \mathbb{Z}^+$. For $x \in \mathbb{Z}_2^k$ and $\mathcal{I}\subseteq [k] $ we have $x[\mathcal{I}]=[x_i:i\in \mathcal{I}]$, $x^i=[x_1,\ldots,x_i]$ and $x_j^i=[x_j,\ldots,x_i]$ for $j<i$. The set $\mathcal{A}^{\setC}$ denotes the complement of the set $\mathcal{A}$. The uniform distribution on an arbitrary random variable $X$ is denoted by $\overline{P}_X$. For distributions $P$ and $Q$ over the same alphabet $\mathcal{X}$, the variational distance is defined by $\delta(P,Q):=\tfrac{1}{2}\sum_{x\in \mathcal{X}} \left| P(x)-Q(x) \right|$. 
The notation $X\markovDavid Y \markovDavid Z$ means that the random variables $X,Y,Z$ form a Markov chain in the given order. 

In this setup we consider a discrete memoryless wiretap channel (DM-WTC) $\W:\mathcal{X}\to \mathcal{Y}\times \mathcal{Z}$, which is characterized by its transition probability distribution $P_{Y,Z|X}$. We assume that the variable $X$ belongs to Alice, $Y$ to Bob and $Z$ to Eve.

According to K\"orner and Marton \cite{korner77}, a DM-WTC $\W:\mathcal{X}\to \mathcal{Y}\times \mathcal{Z}$ is termed \emph{more capable} if $\I{X}{Y} \geq \I{X}{Z}$ for every possible distribution on $X$. The channel $\W$ is termed \emph{less noisy} if $\I{U}{Y} \geq \I{U}{Z}$ for every possible distribution on $(U,X)$ where $U$ has finite support and $U\markovDavid X \markovDavid (Y,Z)$ form a Markov chain. If $X\markovDavid Y \markovDavid Z$ form a Markov chain, $\W$ is called \emph{degraded}. It has been shown \cite{korner77} that being more capable is a strictly weaker condition than being less noisy, which is a strictly weaker condition than being degraded. Hence, having a DM-WTC $\W$ which is degraded implies that $\W$ is less noisy, which again implies that $\W$ is also more capable.

\subsection{Polarization Phenomenon} \label{sec:polarization}
Let $X^N$ be a vector whose entries are i.i.d.\ Bernoulli($p$) distributed for $p\in[0,1]$ and  $N=2^n$ where $n\in \mathbb{Z}^+$. Then define $U^N = G_N X^N$, where $G_N$ denotes the polarization (or polar) transform which can be represented by the matrix
\begin{equation}
G_N := \begin{pmatrix}
1 & 1\\
0 & 1
\end{pmatrix}^{\!\! \otimes  \log N}, \label{eq:polarTrafo}
\end{equation}
where $A^{\otimes k}$ denotes the $k$th Kronecker power of an arbitrary matrix $A$. Furthermore, let $Y^N = \W^N X^N$, where $\W^N$ denotes $N$ independent uses of a DMC $\W:\mathcal{X}\to \mathcal{Y}$. For $\epsilon \in (0,1)$ we may define the two sets
\begin{align}
\mathcal{R}_{\epsilon}^N(X|Y)&:= \left \lbrace i \in[N]: \Hc{U_i}{U^{i-1},Y^N}\geq 1-\epsilon \right \rbrace \quad \textnormal{and} \label{eq:Rset}\\
\mathcal{D}_{\epsilon}^N(X|Y)&:= \left \lbrace i \in[N]: \Hc{U_i}{U^{i-1},Y^N}\leq \epsilon \right \rbrace. \label{eq:Dset}
\end{align}
The former consists of outputs $U_j$ which are essentially uniformly random, even given all previous outputs $U^{j-1}$ as well as $Y^N$, while the latter set consists of the essentially deterministic outputs. The polarization phenomenon is that essentially all outputs are in one of these two subsets, and their sizes are given by the conditional entropy of the input $X$ given $Y$.  

\begin{mythm}[Polarization Phenomenon \cite{arikan09,arikan10}] \label{thm:polPheno}
For any $\epsilon \in (0,1)$
\begin{equation}
\left| \mathcal{R}_\epsilon^N(X|Y) \right|= N \Hc{X}{Y} - o(N) \quad \textnormal{and} \quad \left| \mathcal{D}_\epsilon^N(X|Y) \right|=N \left(1-\Hc{X}{Y}\right) - o(N).
\end{equation}
\end{mythm}
Based on this theorem it is possible to construct a family of linear error correcting codes, called \emph{polar codes}, that have several desirable attributes \cite{arikan09,sasoglu09,arikantelatar09,honda12}:\ they provably achieve the capacity of any DMC; they have an encoding and decoding complexity that is essentially linear in the blocklength $N$; the error probability decays exponentially in the square root of the blocklength. 

Correlated sequences of binary random variables may be polarized using a multilevel construction, as shown in \cite{sasoglu09}.\footnote{An alternative approach is given in \cite{abbe11_2,sahebi11}, where the polarization phenomenon has been generalized for arbitrary finite fields. We will however focus on the multilevel construction in this paper.} Given $M$ i.i.d.\ instances of a sequence $X=(X_{(1)},X_{(2)},\dots X_{(K)})$ and possibly a correlated random variable $Y$, the basic idea is to first polarize $X_{(1)}^M$ relative to $Y^M$, then treat $X_{(1)}^MY^M$ as side information in polarizing $X_{(2)}^M$, and so on. More precisely, defining $U^M_{(j)}=G_MX^M_{(j)}$ for $j=1,\dots,K$, we may define the random and deterministic sets for each $j$ as 
\begin{align}
	\mathcal R_{\epsilon,(j)}^M(X_{(j)}|X_{(j-1)},\cdots, X_{(1)},Y)&=\{i\in[M]: \Hc{U_{(j),i}}{U_{(j)}^{i-1},X_{(j-1)}^M,\cdots, X_{(1)}^M,Y^M}\geq 1-\epsilon\},\\
	\mathcal D_{\epsilon,(j)}^M(X_{(j)}|X_{(j-1)},\cdots, X_{(1)},Y)&=\{i\in[M]: \Hc{U_{(j),i}}{U_{(j)}^{i-1},X_{(j-1)}^M,\cdots, X_{(1)}^M,Y^M}\leq \epsilon\}.
\end{align}
In principle we could choose different $\epsilon$ parameters for each $j$, but this will not be necessary here. Now, Theorem~\ref{thm:polPheno} applies to the random and deterministic sets for every $j$. The sets $\mathcal R_{\epsilon}^{M}(X|Y)= \{ \mathcal R_{\epsilon,(j)}^M(X_{(j)}|X_{(j-1)}, \ldots,X_{(1)},Y) \}_{j=1}^K$ and $\mathcal D_{\epsilon}^{M}(X|Y)=\{ \mathcal D_{\epsilon,(j)}^M(X_{(j)}|X_{(j-1)},\ldots,X_{(1)},Y) \}_{j=1}^K$ have sizes given by  
\begin{align}
	|\mathcal{R}_{\epsilon}^{M}(X|Y)| &= \sum_{j=1}^{K} \left| \mathcal R_{\epsilon,(j)}^M(X_{(j)}|X_{(j-1)}, \ldots,X_{(1)},Y) \right|\\
	&=\sum_{j=1}^{K} M  \Hc{X_{(j)}}{X_{(1)},\dots,X_{(j-1)},Y}-o(M)\\
                                         &=M \Hc{X}{Y}-o(KM),
\end{align}
and 
\begin{align}
	|\mathcal{D}_{\epsilon}^{M}(X|Y)| &= \sum_{j=1}^{K} \left| \mathcal D_{\epsilon,(j)}^M(X_{(j)}|X_{(j-1)},\ldots,X_{(1)},Y) \right|\\
	&=\sum_{j=1}^{K} M\left(1-\Hc{X_{(j)}}{X_{(1)},\dots,X_{(j-1)},Y}\right)-o(M)\\
                                         &=M\left(K-\Hc{X}{Y}\right)-o(KM).
\end{align}
In the following we will make use of both the polarization phenomenon in its original form, Theorem~\ref{thm:polPheno}, and the multilevel extension. To simplify the presentation, we denote by $\widetilde G_{M}^K$ the $K$ parallel applications of $G_M$ to the $K$ random variables $X^M_{(j)}$.

\subsection{One-Way Secret-Key Agreement} \label{subsec:SKA}
At the start of the one-way secret-key agreement protocol, Alice, Bob, and Eve share $N=2^n$, $n\in\mathbb{Z}^+$ i.i.d.\ copies $(X^N,Y^N,Z^N)$ of a triple of correlated random variables $(X,Y,Z)$ which take values in discrete but otherwise arbitrary alphabets $\mathcal X$, $\mathcal Y$, $\mathcal Z$.\footnote{The correlation of the random variables $(X,Y,Z)$ is described by its joint probability distribution $P_{X,Y,Z}$.} %Here we assume $X$ takes values in $\mathcal X=\{0,1\}$ for simplicity, while the alphabets $\mathcal Y$ and $\mathcal Z$ are discrete, but otherwise arbitrary.  

Alice starts the protocol by performing an operation $\tau_A:\mathcal X^N\rightarrow(\mathcal S^J,\mathcal C)$ on $X^N$ which outputs both her secret key $S_A^{J}\in\mathcal{S}^{J}$ for $\mathcal{S}=\{0,1\}$ and an additional random variable $C\in\mathcal C$ which she transmits to Bob over an insecure but noiseless public channel. Bob then performs an operation $\tau_B:(\mathcal Y^N,\mathcal C)\rightarrow\mathcal S^{J}$ on $Y^N$ and the information $C$ he received from Alice to obtain a vector $S_B^{J}\in \mathcal{S}^J$; his secret key. The secret-key thus produced should be reliable, i.e.,\ satisfy the
\begin{equation}
\textnormal{reliability condition:}\quad \lim \limits_{N\to \infty}\Prv{S_A^{J}\ne S_B^{J}}=0,\label{eq:reliability}
\end{equation}
and secure, i.e.,\ satisfy the
\begin{equation}
\textnormal{(strong) secrecy condition:}\quad \lim_{N\rightarrow \infty} \norm{P_{S_A^{J}, Z^N, C}-\overline P_{S_A^{J}}\times P_{Z^N,C}}_1=0,%\leq \epsilon_N \quad \textnormal{for} \quad \lim \limits_{N\to \infty} \epsilon_N =0. 
\label{eq:modern_secrecy}
\end{equation}
where $\overline{P}_X$ denotes the uniform distribution on random variable $X$.

Historically, secrecy was first characterized by a (weak) secrecy condition of the form
 \begin{equation}
 \lim \limits_{N \to \infty} \frac{1}{N}\I{S_A^{J}}{Z^N,C}=0. \label{eq:weak}
\end{equation}
Maurer and Wolf showed that \eqref{eq:weak} is not a sufficient secrecy criterion \cite{maurer94,maurer00} and introduced the strong secrecy condition
\begin{equation}
 \lim \limits_{N \to \infty} \I{S_A^{J}}{Z^N,C}=0, \label{eq:secrecy}
\end{equation}
where in addition it is required that the key is uniformly distributed, i.e.,
\begin{equation}
\lim \limits_{N\to\infty} \TD{P_{S_A^J}}{\overline P_{S_A^J}}=0.
\end{equation}
In recent years, the strong secrecy condition \eqref{eq:secrecy} has often been replaced by \eqref{eq:modern_secrecy}, since (half) the $L_1$ distance directly bounds the probability of distinguishing the actual key produced by the protocol with an ideal key. This operational interpretation is particularly helpful in the finite blocklength regime. In the limit $N\to \infty$, the two secrecy conditions \eqref{eq:modern_secrecy} and \eqref{eq:secrecy} are equivalent, which can be shown using Pinskser's and Fano's inequalities.

Since having weak secrecy is not sufficient, we will only consider strong secrecy in this paper. It has been proven that each secret-key agreement protocol which achieves weak secrecy can be transformed into a strongly secure protocol \cite{maurer00}. However, it is not clear whether the resulting protocol is guaranteed to be practically efficient.

For \emph{one-way} communication, Csisz\'ar and K\"orner \cite{csiszar78} and later Ahlswede and Csisz\'ar \cite{ahlswede93} showed that the optimal rate $R:=\lim_{N\rightarrow \infty}\frac {J}N$ of generating a secret key satisfying \eqref{eq:reliability} and \eqref{eq:secrecy}, called the \emph{secret-key rate} $\SK{X}{Y}{Z}$, is characterized by a closed single-letter formula.
\begin{mythm}[\cite{csiszar78,ahlswede93}] \label{thm:skaRate} 
For triples $(X,Y,Z)$ described by $P_{X,Y,Z}$ as explained above,
\begin{equation} 
\SK{X}{Y}{Z}=\left \lbrace\begin{array}{rl}
\max\limits_{P_{U,V}} & \Hc{U}{Z,V}-\Hc{U}{Y,V}\\
\mathrm{s.t.} & V\markovDavid U \markovDavid X \markovDavid (Y,Z), \\
& |\mathcal{V}| \leq |\mathcal{X}|, \, |\mathcal{U}| \leq |\mathcal{X}|^2.
\end{array} \right. \label{eq:highestRate}
\end{equation}
\end{mythm}

The expression for the one-way secret-key rate given in Theorem~\ref{thm:skaRate} can be simplified if one makes additional assumptions about $P_{X,Y,Z}$.
\begin{mycor}
For $P_{X,Y,Z}$ such that the induced DM-WTC $\W$ described by $P_{Y,Z|X}$ is more capable,
\begin{equation} 
\SK{X}{Y}{Z}=\left \lbrace\begin{array}{rl}
\max\limits_{P_V} & \Hc{X}{Z,V}-\Hc{X}{Y,V}\\
\mathrm{s.t.} & V\markovDavid X \markovDavid (Y,Z),  \label{eq:skrMC}\\
& |\mathcal{V}| \leq |\mathcal{X}|.
\end{array} \right.
\end{equation}
\end{mycor}
\begin{proof}In terms of the mutual information, we have 
\begin{align}
	\Hc{U}{Z,V}-\Hc{U}{Y,V} &= \Ic{U}{Y}{V}-\Ic{U}{Z}{V}\\
	&=\Ic{X,U}{Y}{V}-\Ic{X,U}{Z}{V}-\left(\Ic{X}{Y}{U,V}-\Ic{X}{Z}{U,V}\right)\\
	&\leq \Ic{X,U}{Y}{V}-\Ic{X,U}{Z}{V}\\
	&=\Ic{X}{Y}{V}-\Ic{X}{Z}{V},
\end{align}
using the chain rule, the more capable condition, and the Markov chain properties, respectively. Thus, the maximum in $S_{\to}(X;Y|Z)$ can be achieved when omitting $U$. 
\end{proof}

\begin{mycor} \label{cor:LN_SKA}
For $P_{X,Y,Z}$ such that the induced DM-WTC $\W$ described by $P_{Y,Z|X}$ is less noisy,
\begin{equation} 
\SK{X}{Y}{Z}=  \Hc{X}{Z}-\Hc{X}{Y}. \label{eq:skaLessNoisy}
\end{equation}
\end{mycor}
\begin{proof}
Since $\W$ being less noisy implies $\W$ being more capable, we know that the one-way secret key rate is given by \eqref{eq:skrMC}. Using the chain rule we obtain
\begin{align}
	\Hc{X}{Z,V}-\Hc{X}{Y,V}&=\Ic{X}{Y}{V}-\Ic{X}{Z}{V}\\ 
	&= \I{X,V}{Y}-\I{X,V}{Z}-\I{V}{Y}+\I{V}{Z}\\
&=\I{X}{Y}-\I{X}{Z}-\left(\I{V}{Y}-\I{V}{Z} \right) \label{eq:MC}\\
&\leq \I{X}{Y}-\I{X}{Z}. \label{eq:IneqLN}
\end{align}
Equation~\eqref{eq:MC} follows from the chain rule and the Markov chain condition. The inequality uses the assumption of being less noisy.
\end{proof}
Note that \eqref{eq:skaLessNoisy} is also equal to the one-way secret-key rate for the case where $\W$ is degraded, as this implies $\W$ being less noisy.
The proof of Theorem~\ref{thm:skaRate} does not imply that there exists an \emph{efficient} one-way secret-key agreement protocol. 
A computationally efficient scheme was constructed in \cite{holenstein05}, but is not known to be practically efficient.\footnote{As defined in Section~\ref{sec:intro}, we call a scheme practically efficient if its computational complexity is essentially linear in the blocklength.} 

For key agreement with two-way communication, no formula comparable to \eqref{eq:highestRate} for the optimal rate is known. However, it has been shown that the two-way secret-key rate is strictly larger than the one-way secret-key rate. It is also known that the \emph{intrinsic information} $I(X;Y\!\!\downarrow \!Z):=\min_{P_{Z'|Z}}\Ic{X}{Y}{Z'}$ is an upper bound on $S(X;Y|Z)$, but is not tight \cite{ahlswede93,maurer99,renner03}.

%%%%%%%%%%%%%%%%%%%%%%%%%%%%%%%%%%%%%%%%%%%%%%%%%%%%%
\subsection{Private Channel Coding} \label{sec:introPCC}
Private channel coding over a wiretap channel is closely related to the task of one-way secret-key agreement from common randomness (cf.\ Section~\ref{sec:relation}). Here Alice would like to transmit a message $M^J \in \mathcal{M}^J$ privately to Bob. The messages can be distributed according to some arbitrary distribution $P_{M^J}$. To do so, she first encodes the message by computing $X^N=\enc(M^J)$ for some encoding function $\enc:\mathcal{M}^J \to \mathcal{X}^N$ and then sends $X^N$ over the wiretap channel to Bob (and to Eve), which is represented by $(Y^N,Z^N)=\W^N X^N$. Bob next decodes the received message to obtain a guess for Alice's message $\hat M^J = \dec(Y^N)$ for some decoding function $\dec:\mathcal{Y}^N \to \mathcal{M}^J$. As in secret-key agreement, the private channel coding scheme should be reliable, i.e.,\ satisfy the
\begin{equation}
\textnormal{reliability condition:}\quad \lim \limits_{J\to \infty}\Prv{M^J\ne \hat M^J}=0 \label{eq:reliabilityPCC}
\end{equation}
and (strongly) secure, i.e.,\ satisfy the
\begin{equation}
\textnormal{(strong) secrecy condition:}\quad \lim_{J\rightarrow \infty}\norm{P_{M^J,Z^N,C}-P_{M^J}\times P_{Z^N,C}}_1=0.%\leq \epsilon_J \quad \textnormal{for} \quad \lim \limits_{J\to \infty}\epsilon_J =0. 
\label{eq:secrecyPCC}
\end{equation}
The variable $C$ denotes any additional information made public by the protocol. 

As mentioned in Section~\ref{subsec:SKA}, in the limit $J\to \infty$ this strong secrecy condition is  equivalent to the historically older (strong) secrecy condition \begin{equation}
\lim \limits_{J \to \infty} \I{M^J}{Z^N,C}=0. 
\end{equation} 
%The variable $C_N$ denotes additional one-way communication from Alice to Bob, carried out over a public noiseless channel.\footnote{We will show later that our private channel coding scheme does not require additional one-way communication.} 
The highest achievable rate $R:=\lim_{N\rightarrow \infty}\frac JN$ fulfilling \eqref{eq:reliabilityPCC} and \eqref{eq:secrecyPCC} is called the \emph{secrecy capacity}. 

Csisz\'ar and K\"orner showed \cite[Corollary 2]{csiszar78} that there exists a single-letter formula for the secrecy capacity.\footnote{Maurer and Wolf showed that the single-letter formula remains valid considering strong secrecy \cite{maurer00}.}
\begin{mythm}[\cite{csiszar78}]
For an arbitrary DM WTC $\W$ as introduced above,
\begin{equation} \label{eq:secrecyCapacity}
C_s=\left \lbrace \begin{array}{rl}
\max \limits_{P_{V,X}} & \Hc{V}{Z}-\Hc{V}{Y}\\
\mathrm{s.t.} & V \markovDavid X \markovDavid (Y,Z), \\
&|\mathcal{V}|\leq |\mathcal{X}|. 
\end{array} \right.
\end{equation}
\end{mythm}
This expression can be simplified using additional assumptions about $\W$.
\begin{mycor}[\cite{korner77}] \label{cor:PCCmoreCapable}
If $\W$ is more capable,
\begin{equation}
C_s= \Hc{X}{Z}-\Hc{X}{Y}. \label{eq:capacityCapable}
\end{equation}
\end{mycor}
\begin{proof}
A proof can be found in \cite{korner77} or \cite[Section 22.1]{elgamal12}.
\end{proof}

\subsection{Previous Work and Our Contributions} \label{sec:relation}

In Section~\ref{sec:ska}, we present a one-way secret-key agreement scheme based on polar codes that achieves the secret-key rate, is strongly secure, reliable and whose implementation is practically efficient, with complexity $O(N \log N)$ for blocklength $N$.   
Our protocol improves previous efficient secret-key constructions \cite{abbeITW}, where only weak secrecy could be proven and where the eavesdropper has no prior knowledge and/or degradability assumptions are required. However, we are not able to give an efficient algorithm for code construction.

In Section~\ref{sec:pcc}, we introduce a coding scheme based on polar codes that provably achieves the secrecy capacity for arbitrary discrete memoryless wiretap channels. We show that the complexity of the  encoding and decoding operations is $O(N\log N)$ for blocklength $N$. Our scheme improves previous work on practically efficient private channel coding at the optimal rate \cite{vardy11}, where only weak secrecy could be proven under the additional assumption that the channel $\W$ is degraded.\footnote{Note that Mahdavifar and Vardy showed that their scheme achieves strong secrecy if the channel to Eve (induced from $\W$) is noiseless. Otherwise their scheme is not provably reliable \cite{vardy11}.} Recently, Bellare et al. introduced an efficient coding scheme that is strongly secure and achieves the secrecy capacity for binary symmetric wiretap channels \cite{bellare12}.\footnote{They claim that their scheme works for a large class of wiretap channels. However, this class has not been characterized precisely so far. It is therefore not clear wether their scheme requires for example degradability assumptions. Note that to obtain strong secrecy for an arbitrarily distributed message, it is required that the wiretap channel is symmetric \cite[Lemma14]{bellare12}.}  Several other constructions of private channel coding schemes have been reported \cite{andersson10,hof10,gamal10}, but all achieve only weak secrecy.

The tasks of one-way secret-key agreement and private channel coding explained in the previous two subsections are closely related. Maurer showed how a one-way secret-key agreement can be derived from a private channel coding scenario \cite{maurer93}. More precisely, he showed how to obtain the common randomness needed for one-way secret-key agreement by constructing a ``virtual'' degraded wiretap channel from Alice to Bob. This approach can be used to obtain the one-way secret-key rate from the secrecy capacity result in the wiretap channel scenario \cite[Section 22.4.3]{elgamal12}.

One of the main advantages of the two schemes introduced in this paper is that they are both practically efficient. However, even given a practically efficient private coding scheme, it is not known that Maurer's construction will yield a practically efficient scheme for secret key agreement. 
For this reason, as well as simplicity of presentation, we treat the one-way secret-key agreement and the private channel coding problem separately in the two sections to follow.

%%%%%%%%%%%%%%%%%%%%%%%%%%%%%%%%%%%
%%%%%%%%%%%%%%%%%%%%%%%%%%%%%%%%%%%

\section{One-Way Secret-Key Agreement Scheme}\label{sec:ska}
Our key agreement protocol is a concatenation of two subprotocols, an inner and an outer layer, as depicted in Figure~\ref{fig:SKA}. The protocol operates on blocks of $N$ i.i.d.\ triples $(X,Y,Z)$, which are divided into $M$ sub-blocks of size $L$ for input to the inner layer. In the following we assume $\mathcal{X}=\{0,1\}$, which however is only for convenience; the techniques of \cite{sasoglu09} and \cite{karzand10} can be used to generalize the schemes to discrete memoryless wiretap channels with arbitrary input size.

The task of the inner layer is to perform \emph{information reconciliation} and that of the outer layer is to perform \emph{privacy amplification}. %This approach is similar to the one we introduced in \cite{concEntagDist} for entanglement distillation in the quantum-mechanical setting. \jmr{and many other approaches...} 
Information reconciliation refers to the process of carrying out error correction to ensure that Alice and Bob obtain a shared bit string, and here we only allow communication from Alice to Bob for this purpose. On the other hand, privacy amplification refers to the process of distilling from Alice's and Bob's shared bit string a smaller set of bits whose correlation with the information available to Eve is below a desired threshold.

Each subprotocol in our scheme is based on the polarization phenomenon. For information reconciliation of Alice's random variable $X^L$ relative to Bob's information $Y^L$, Alice applies a polar transformation to $X^L$ and forwards the bits of the complement of the deterministic set $\mathcal D_{\epsilon_1}^L(X|Y)$ to Bob over a insecure public channel, which enables him to recover $X^L$ using the standard polar decoder~\cite{arikan09}. Her remaining information is then fed into a multilevel polar transformation and the bits of the random set are kept as the secret key.

Let us now define the protocol more precisely.
For $L=2^\ell$, $\ell \in \mathbb{Z}^+$, let $V^L=G_L X^L$ where $G_L$ is as defined in \eqref{eq:polarTrafo}. For $\epsilon_1>0$, we define 
\begin{align}
\mathcal E_K:=\mathcal D_{\epsilon_1}^L(X|Y), \label{eq:Eset}
\end{align}
with $K:=|\mathcal D_{\epsilon_1}^L(X|Y)|$. Then, let $T_{(j)}=V^L[\mathcal{E}_K]_j$ for $j=1,\dots,K$ and $C_{(j)}=V^L[\mathcal{E}_K^c]_j$ for $j=1,\dots,L-K$ so that $T=(T_{(1)},\dots,T_{(K)})$ and $C=(C_{(1)},\dots,C_{(L-K)})$.  
For $\epsilon_2>0$ and $U^{M}_{(j)}= G_{M}T_{(j)}^M$ for $j=1,\dots K$ (or, more briefly, $U^M=\widetilde G_M^KT^M$),  we define 
\begin{align}
\mathcal F_J:=\mathcal R_{\epsilon_2}^M(T|CZ^L), \label{eq:Fset}
\end{align}
with $J:=|R_{\epsilon_2}^M(T|CZ^L)|$.

 \begin{table}[!htb]
\centering 
\begin{tabular}{c}
  \Xhline{3\arrayrulewidth}  \hspace{1mm} \vspace{-3mm}\\ 
\hspace{45mm}{\bf{Protocol 1:}} \hypertarget{protocol:ska}{One-way secret-key agreement}\hspace{45mm} \\ \vspace{-3mm} \\ \hline 
\end{tabular}
  \begin{tabular}{l l}
  {\bf{Given:}} & Index sets $\mathcal{E}_K$ and $\mathcal{F}_J$ (code construction)\\
   {\bf{Notation:}} & Alice's input: $x^N \in \mathbb{Z}_2^N$ (a realization of $X^N$)\\
 & Bob's / Eve's input: $(y^N,z^N)$ (realizations of $Y^N$ and $Z^N$) \\
 & Alice's output: $s_A^J$\\
 & Bob's output: $s_B^J$\\ \\
{\bf Step 1: } & Alice computes $v_{i+1}^{i+L}=G_{L}x_{i+1}^{i+L}$ for all $i\in\{0,L,2L,\ldots,(M-1)L\}$.\\
{\bf Step 2: } & Alice computes $t_i=v_{i+1}^{i+L}[\mathcal{E}_K]$ for all $i\in\{0,L,2L,\ldots,(M-1)L\}$.\\
{\bf Step 3: } & Alice sends $c_i=v^{i+L}_{i+1}[\mathcal{E}_K^\setC]$ for all $i\in\{0,L,2L,\ldots,(M-1)L\}$ over a public channel to Bob.\\
{\bf Step 4: } & Alice computes $u^{M}=\widetilde G_{M}^K  t^{M}$ and obtains $s_A^{J}=u^{M}[\mathcal{F}_J]$.\footnotemark[\value{footnote}]\\ 
{\bf Step 5: } & Bob applies the standard polar decoder \cite{arikan09,honda12} to $(c_i,y_{i+1}^{i+L})$ to obtain $\hat v^{i+L}_{i+1}$ and \\
&${\hat t}_i=\hat v_{i+1}^{i+L}[\mathcal{E}_K]$, for $i \in \{0,L,2L,\ldots,(M-1)L\}$.\\
{\bf Step 6: } & Bob computes $\hat u^{M}=\tilde G_{M}^K {t}^{M}$ and obtains $s_B^J= \hat u^{M}[\mathcal{F}_J]$.\\
\vspace{-10mm}
  \end{tabular}
\begin{tabular}{c}
\hspace{43mm} \phantom{ {\bf{Protocol 1:}} One-way secret-key agreement}\hspace{45mm} \\ \vspace{-2.5mm} \\\Xhline{3\arrayrulewidth}
\end{tabular}
\end{table}
 \addtocounter{footnote}{-1}
\footnotetext{The expression $u^M[\mathcal F_J]$ is an abuse of notation, as $\mathcal F_J$ is not a subset of [M]. The expression should be understood to be the union of the random bits of $u_{(j)}^M$, for all $j=1,\dots,K$, as in the definition of $\mathcal R^M_{\epsilon_2}(T|CZ^L)$.\label{Footnote:Protocol1}}
%\end{savenotes}

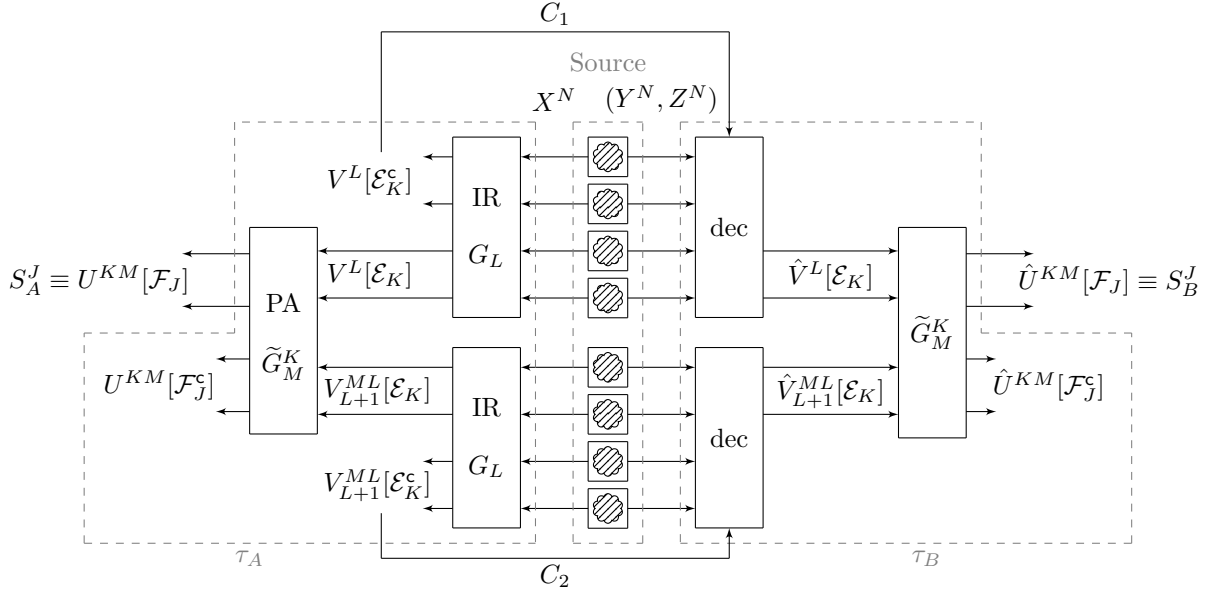
\begin{figure}[!htb]
%\tikzsetnextfilename{Extractor}
\centering
\input{SKAscheme.tex}
\caption{\small The secret-key agreement scheme for the setup $N=8$, $L=4$, $M=2$, $K=2$, and $J=2$. We consider a source that produces $N$ i.i.d.\ copies $(X^N,Y^N,Z^N)$ of a triple of correlated random variables $(X,Y,Z)$. Alice performs the operation $\tau_A$, sends $(V^L[\mathcal{E}_K^\setC])^M$ over a public channel to Bob and obtains $S_A^J$, her secret key. Bob then performs the operation $\tau_B$ which results in his secret key $S_B^J$.}
\label{fig:SKA}
\end{figure}

%%%%%%%%%%%%%%%%%%%%%%%%%%%%%%%%%%%
%%%%%%%%%%%%%%%%%%%%%%%%%%%%%%%%%%%

\subsection{Rate, Reliability, Secrecy, and Efficiency} \label{sec:results}
\begin{mythm} \label{thm:main}
Protocol~\hyperlink{protocol:ska}{1} allows  Alice and Bob to generate a secret key $S_A^J$ respecitvely $S_B^J$ using public one-way communication $C^M$%\footnote{Recall that the public communication $C_N$ is assumed to take place over a noiseless channel.} 
 such that for $\beta < \tfrac{1}{2}$ :
\begin{alignat}{2}
&\textnormal{Reliability:}\quad&&\Prv{S_A^J \ne S_B^J}= O\!\left(M 2^{-L^{\beta}}\right) \label{eq:reliability2}\\
&\textnormal{Secrecy:}&&\norm{P_{S_A^{J}, Z^N, C}-\overline P_{S_A^{J}}\times P_{Z^N,C}}_1= O\!\left(\sqrt{N} 2^{-\frac{N^{\beta}}{2}}\right) \\
&\textnormal{Rate:}&&R:=\frac JN= \Hc{X}{Z}-\frac{1}{L}\Hc{V^L[\mathcal{E}_K^\setC]}{Z^L}-\frac{o(N)}{N}.
\end{alignat}
All operations by both parties may be performed in $O(N \log N)$ steps. 
\end{mythm}
\begin{proof}
The reliability of Alice's and Bob's key follows from the standard polar decoder error probability and the union bound. Each instance of the decoding algorithm employed by Bob has an error probability which scales as $O(2^{-L^{\beta}})$ for $\beta < \tfrac{1}{2}$ \cite{arikan10}; application of the union bound gives the prefactor $M$. 

To prove the secrecy statement requires more effort. Using Pinsker's inequality we obtain
\begin{align}
\TD{P_{S_A^J,Z^N,C^M}}{\overline P_{S_A^J} \times P_{Z^N,C^M}} &\leq  \sqrt{\tfrac{\ln 2}{2} \KL{P_{S_A^J,Z^N,C^M}}{\overline P_{S_A^J} \times P_{Z^N,C^M}}} \\
 &= \sqrt{\tfrac{\ln 2}{2} \left(J-\Hc{S_A^J}{Z^N,C^M} \right)}, \label{eq:notdo}
\end{align}
where the last step uses the chain rule for relative entropies and that $\overline P_{S_A^J}$ denotes the uniform distribution. We can simplify the conditional entropy expression using the chain rule
\begin{align}
&\Hc{S_A^J}{Z^N,C^M} \nonumber\\
&\hspace{8mm}= \Hc{U^M[\mathcal{F}_J]}{Z^N,(V^L[\mathcal{E}_K^{\setC}])^M}\\
&\hspace{8mm}= \sum_{j=1}^K \Hc{U_{(j)}^M[\mathcal{F}_{(j)}]}{U_{(1)}^M[\mathcal{F}_{(1)}],\ldots,U_{(j-1)}^M[\mathcal{F}_{(j-1)}],Z^N,(V^L[\mathcal{E}_K^{\setC}])^M}\\
&\hspace{8mm}=\sum_{j=1}^K \sum_{i=1}^{\left|\mathcal{F}_{(i)} \right|}\Hc{U_{(j)}^M[\mathcal{F}_{(j)}]_i}{U_{(j)}^M[\mathcal{F}_{(j)}]^{i-1},U_{(1)}^M[\mathcal{F}_{(1)}],\ldots,U_{(j-1)}^M[\mathcal{F}_{(j-1)}],Z^N,(V^L[\mathcal{E}_K^{\setC}])^M}\\
&\hspace{8mm}\geq \sum_{j=1}^K \sum_{i \in \mathcal{F}_j} \Hc{U_{(j)i}}{U_{(j)}^{i-1},U_{(1)}^M[\mathcal{F}_{(1)}], \ldots, U_{(j-1)}^M[\mathcal{F}_{(j-1})],Z^N,(V^L[\mathcal{E}_K^{\setC}])^M} \\
&\hspace{8mm}\geq J\left(1-\epsilon_2\right),
\end{align}
where the first inequality uses the fact that that conditioning reduces entropy and the second inequality follows by the definition of $\mathcal{F}_J$. Recall that we are using the notation introduced in Section~\ref{sec:polarization}. For $\mathcal{F}_J$ as defined in \eqref{eq:Fset}, we have $\mathcal{F}_J=\left \lbrace \mathcal{F}_{(j)} \right \rbrace_{j=1}^{K}$ where $\mathcal{F}_{(j)}=\mathcal{R}_{\epsilon_2}^M\left(T_{(j)}\left| T_{(j-1)},\ldots,T_{(1)},C,Z^L \right.\right)$.
The polarization phenomenon, Theorem~\ref{thm:polPheno}, implies $J=O(N)$, which together with \eqref{eq:notdo} proves the secrecy statement of Theorem~\ref{thm:main}, since $\epsilon_2=O(2^{-N^\beta})$ for $\beta<\tfrac{1}{2}$.

%Using the definition of the mutual information and the chain rule we obtain
%\begin{align}
%\I{S_A^J}{Z^N, C^M} &=\I{U^{M}[\mathcal{F}_J]}{Z^N, V^L[\mathcal{E}_K^\setC]^M}\\
%&=\Hh{U^{M}[\mathcal{F}_J]}-\Hc{U^{M}[\mathcal{F}_J]}{Z^N,V^L[\mathcal{E}_K^\setC]^M}\\
%&=\Hh{U^{M}[\mathcal{F}_J]}-\sum_{i=1}^J \Hc{U^{M}[\mathcal{F}_J]_i}{U^{M}[\mathcal{F}_J]^{i-1},Z^N,V^L[\mathcal{E}_K^\setC]^M}\\
%&\leq \Hh{U^{M}[\mathcal{F}_J]}-\sum_{i\in \mathcal{F}_J} \Hc{U_i}{U^{i-1},Z^N,V^L[\mathcal{E}_K^\setC]^M}\\ 
%&\leq \Hh{U^{M}[\mathcal{F}_J]} -J\left(1-\epsilon_2\right)\\
%&\leq J\epsilon_2,
%\end{align}
%where the first inequality uses the fact that that conditioning reduces entropy and the second inequality follows by the definition of $\mathcal{F}_J$. The final inequality uses the upper bound of the entropy in terms of the alphabet size. 
%Using Pinsker's inquality we can prove the second requirement of the secrecy condition
%\begin{align}
%\TD{P_{S_A^J}}{\overline P_{S_A^J}} &\leq \sqrt{\tfrac{\ln 2}{2} \KL{P_{S_A^J}}{\overline P_{S_A^J}} }\\
%& =\sqrt{\tfrac{\ln 2}{2} \left(J-\Hh{S_A^J} \right)}\\
%& =\sqrt{\tfrac{\ln 2}{2} \left(J-\Hh{U^{M}[\mathcal{F}_J]} \right)}\\
%&\leq \sqrt{\tfrac{\ln 2}{2}J \epsilon_2},
%\end{align}
%where the final inquality follows by definition of $\mathcal{F}_J$.
%The polarization phenomenon, Theorem~\ref{thm:polPheno}, implies $J=O(N)$, which proves the secrecy statement of Theorem~\ref{thm:main}, since $\epsilon_2=O(2^{-N^\beta})$ for $\beta<\tfrac{1}{2}$.

The rate of the scheme is
\begin{align}
R&=\frac{\left|\mathcal{F}_J\right|}{N}\\
&=\frac{1}{L}\Hc{V^L[\mathcal{E}_K]}{V^L[\mathcal{E}_K^\setC],Z^L}-\frac{o(N)}{N} \label{eq:rateForLater}\\
&=\frac{1}{L}\left(\Hc{V^L}{Z^L}-\Hc{V^L[\mathcal{E}_K^\setC]}{Z^L} \right)-\frac{o(N)}{N}\\
&=\Hc{X}{Z}-\frac{1}{L}\Hc{V^L[\mathcal{E}_K^\setC]}{Z^L}-\frac{o(N)}{N}, \label{eq:ratefini}
\end{align}
where \eqref{eq:rateForLater} uses the polarization phenomenon stated in Theorem~\ref{thm:polPheno}.

It remains to show that the computational complexity of the scheme is $O(N \log N)$. Alice performs the operation $G_L$ in the first layer $M$ times, each requiring $O(L \log L)$ steps \cite{arikan09}. In the second layer she performs $\tilde G_M^K$, or $K$ parallel instances of $G_M$, requiring  $O(KM\log M)$ total steps. From the polarization phenomenon, we have $K=O(L)$, and thus the complexity of Alice's operations is not worse than $O(N\log N)$. Bob runs $M$ standard polar decoders which can be done in $O(M L \log L)$ complexity \cite{arikan09,honda12}. Bob next performs the polar transform $\widetilde G_M^K$, whose complexity is not worse than $O(N \log N)$ as justified above. Thus, the complexity of Bob's operations is also not worse than $O(N\log N)$. 
\end{proof}

In principle, the two parameters $L$  and $M$ %\footnote{Recall that $L$ denotes the number of input bits per block at the inner layer whereas $M$ denotes the number of blocks at the inner layer.} 
can be chosen freely. However, to maintain the reliability of the scheme (cf.\eqref{eq:reliability2}), $M$ may not grow exponentially fast in $L$. A reasonable choice would be to have both parameters scale comparably fast, i.e.,\ $\frac{M}{L}=O(1)$.

\begin{mycor} \label{cor:rate}
The rate of Protocol~\hyperlink{protocol:ska}{1} given in Theorem~\ref{thm:main} can be bounded as
\begin{equation}
R\geq \max\left \lbrace 0, \Hc{X}{Z}-\Hc{X}{Y}-\frac{o(N)}{N}\right \rbrace.
\end{equation}
\end{mycor}
\begin{proof}
According to \eqref{eq:ratefini} the rate of Protocol~\hyperlink{protocol:ska}{1} is
\begin{align}
R&=\Hc{X}{Z}-\frac{1}{L}\Hc{V^L[\mathcal{E}_K^\setC]}{Z^L}-\frac{o(N)}{N}\\
&\geq \max \left \lbrace 0,\Hc{X}{Z}-\frac{\left|\mathcal{E}_K^\setC\right|}{L}-\frac{o(N)}{N} \right \rbrace \label{eq:tight}\\
&= \max \left \lbrace 0, \Hc{X}{Z}-\Hc{X}{Y}-\frac{o(N)}{N}\right \rbrace, \label{eq:laststepR}
\end{align}
where \eqref{eq:laststepR} uses the polarization phenomenon stated in Theorem~\ref{thm:polPheno}.
\end{proof}

%%%%%%%%%%%%%%%%%%%%%%%%%%%%%%%%%%%%
%%%%%%%%%%%%%%%%%%%%%%%%%%%%%%%%%%%

%%%%%%%%%%%%%%%%%%%%%%%%%%%%%%%%%%%
%%%%%%%%%%%%%%%%%%%%%%%%%%%%%%%%%%%

\subsection{Achieving the Secret-Key Rate}
Theorem~\ref{thm:main} together with Corollaries~\ref{cor:LN_SKA} and~\ref{cor:rate} immediately imply that Protocol~\hyperlink{protocol:ska}{1} achieves the secret-key rate $\SK{X}{Y}{Z}$ if $P_{X,Y,Z}$ is such that the induced DM WTP $\W$ is less noisy.
If we can solve the optimization problem \eqref{eq:highestRate}, i.e.,\ find the optimal auxiliary random variables $V$ and $U$, our one-way secret-key agreement scheme can achieve $\SK{X}{Y}{Z}$ for a general setup. We then make $V$ public, replace $X$ by $U$ and run Protocol~\hyperlink{protocol:ska}{1}.  Note that finding the optimal random variables $V$ and $U$ might be difficult. It has been shown that for certain distributions the optimal random variables $V$ and $U$ can be found analytically \cite{holenstein05}.

Two open problems discussed in Section~\ref{sec:discussion} address the question if Protocol~\hyperlink{protocol:ska}{1} can achieve a rate that is strictly larger than $\max \left \lbrace 0, \Hc{X}{Z}-\Hc{X}{Y} \right \rbrace$ if nothing about the optimal auxiliary random variables $V$ and $U$ is known, i.e.,\ if we run the protocol directly for $X$ without making $V$ public.

\subsection{Code Construction} \label{subsec:SKA_codeConstr}
Before the protocol starts one must construct the code, i.e.\ compute the index sets $\mathcal{E}_K$ and $\mathcal{F}_J$. 
The set $\mathcal{E}_K$ can be computed approximately with a linear-time algorithm introduced in \cite{tal12}, given the distributions $P_X$ and $P_{Y|X}$. 
 Alternatively, Tal and Vardy's older algorithm \cite{talandvardy10} and its adaption to the asymmetric setup \cite{honda12} can be used.

To compute the outer index set $\mathcal{F}_J$ even approximately requires more effort. In principle, we can again use the above algorithms, which require a description of the ``super-source'' seen by the outer layer, i.e.\ the source which outputs the triple of random variables  $(V^L[\mathcal{E}_K],(Y^L, V^L[\mathcal{E}_K^{\setC}]),(Z^L, V^L[\mathcal{E}_K^{\setC}]))$. However, its alphabet size is exponential in $L$, and thus such a direct approach will not be efficient in the overall blocklength $N$. Nonetheless, due to the structure of the inner layer, it is perhaps possible that the method of approximation by limiting the alphabet size~\cite{talandvardy10,tal12} can be extended to this case. 
%index set $\mathcal{F}_J$ can be computed efficiently by applying similar techniques as introduced in \cite{tal12}. 
In particular, a recursive construction motivated by the decoding operation introduced in \cite{concEntagDist} could potentially lead to an efficient computation of the index set $\mathcal{F}_J$.

%%%%%%%%%%%%%%%%%%%%%%%%%%%%%%%%%%%
%%%%%%%%%%%%%%%%%%%%%%%%%%%%%%%%%%%
\section{Private Channel Coding Scheme}\label{sec:pcc}
Our private channel coding scheme is a simple modification of the secret key agreement protocol of the previous section. Again it consists of two layers, an inner layer which ensures transmitted messages can be reliably decoded by the intended receiver, and an outer layer which guarantees privacy from the unintended receiver. The basic idea is to simply run the key agreement scheme in reverse, \emph{inputting} messages to the protocol where secret key bits would be \emph{output} in key agreement. The immediate problem in doing so is that key agreement also produces outputs besides the secret key, so the procedure is not immediately reversible. To overcome this problem, the encoding operations here simulate the random variables output in the key agreement protocol, and then perform the polar transformations $\widetilde G_M^K$ and $G_L$ in reverse.\footnote{As it happens, $G_L$ is its own inverse.}  

The scheme is visualized in Figure~\ref{fig:pcc} and described in detail in Protocol~\hyperlink{protocol:pcc}{2}. Not explicitly shown is the simulation of the bits $U^M[\mathcal F_J]$ at the outer layer and the bits $V^L[\mathcal E^c_K]$ at the inner layer. The outer layer, whose simulated bits are nearly deterministic, makes use of the method described in  \cite[Definition 1]{sutter12}, while the inner layer, whose bits are nearly uniformly-distributed, follows \cite[Section IV]{honda12}. Both proceed by successively sampling from the individual bit distributions given all previous values in the particular block, i.e.,\ constructing $V_j$ by sampling from $P_{V_j|V^{j-1}}$. These distributions can be efficiently constructed, as described in Section~\ref{sec:ccpcc}. 
 
Note that a public channel is used to communicate the information reconciliation information to Bob, enabling reliable decoding. However, it is possible to dispense with the public channel and still achieve the same rate and efficiency properties, as will be discussed in Section~\ref{sec:ccpcc}.

In the following we assume that the message $M^J$ to be transmitted is uniformly distributed over the message set $\mathcal{M}=\left \lbrace 0,1 \right \rbrace^J$. As mentioned in Section~\ref{sec:introPCC}, it may be desirable to have a private coding scheme that works for an arbitrarily distributed message. This can be achieved by assuming that the wiretap channel $\W$ is symmetric---more precisely, by assuming that the two channels $\W_1:\mathcal{X}\to \mathcal{Y}$ and $\W_2:\mathcal{X}\to \mathcal{Z}$ induced by $\W$ are symmetric. We can define a super-channel $\W':\mathcal{T}\to \mathcal{Y}^L \times \mathcal{Z}^L\times \mathcal{C}$ which consists of an inner encoding block and $L$ basic channels $\W$.\footnote{This super-channel is explained in more detail in Section~\ref{sec:diss_LN}.} The super-channel $\W'$ again induces two channels $\W_1':\mathcal{T}\to \mathcal{Y}^L \times \mathcal{C}$ and $\W_2':\mathcal{T}\to \mathcal{Z}^L\times \mathcal{C}$. Ar{\i}kan showed that $\W_1$ respectively $\W_2$ being symmetric implies that $\W_1'$ respectively $\W_2'$ is symmetric \cite[Proposition 13]{arikan09}. It has been shown in \cite[Proposition 3]{vardy11} that for symmetric channels polar codes remain reliable for an arbitrary distribution of the message bits. We thus conclude that if $\W_1$ is assumed to be symmetric, our coding scheme remains reliable for arbitrarily distributed messages. Assuming having a symmetric channel $\W_2$ implies that $\W_2'$ is symmetric which proves that our scheme is strongly secure for arbitrarily distributed messages.\footnote{This can be seen easily by the strong secrecy condition given in \eqref{eq:secrecyPCC} using that $\W_2'$ is symmetric.} 

 \begin{table}[!htb]
\centering 
\begin{tabular}{c}
  \Xhline{3\arrayrulewidth}  \hspace{1mm} \vspace{-3mm}\\ 
\hspace{53mm}{\bf{Protocol 2:}} \hypertarget{protocol:pcc}{Private channel coding}\hspace{53mm} \\ \vspace{-3mm} \\ \hline 
\end{tabular}
 %\addtocounter{footnote}{1}
  \begin{tabular}{l l} 
  \addtocounter{footnote}{1}
  {\bf{Given:}} & Index sets $\mathcal{E}_K$ and $\mathcal{F}_J$ (code construction)\footnotemark[\value{footnote}]\\  
  \addtocounter{footnote}{1}
   {\bf{Notation:}} & Message to be transmitted: $m^J$\\\\
{\bf Outer encoding: } & Let $u^{M}[\mathcal{F}_J]=m^J$\footnotemark[\value{footnote}] and $u^{M}[\mathcal{F}_J^{\setC}]=r^{KM-J}$ where $r^{KM-J}$ is (randomly) generated\\
   & as explained in \cite[Definition 1]{sutter12}. Let $t^{M}=\widetilde G_{M}^K u^{M}$.\\
{\bf Inner encoding: } & For all $i\in \{0,L,\ldots,L(M-1) \}$, Alice does the following: let $\bar v_{i+1}^{i+L}[\mathcal{E}_K]=t_{(i/L) +1}$\\
& and $\bar v_{i+1}^{i+L}[\mathcal{E}_K^{\setC}]=s_{i+1}^{i+L-K}$ where $s_{i+1}^{i+L-K}$ is (randomly) generated as explained in \\
&\cite[Section IV]{honda12}. Send $C_{(i/K)+1}:=s_{i+1}^{i+L-K}$ over a public channel to Bob. Finally,\\
& compute $x_{i+1}^{i+L}=G_L \bar v_{i+1}^{i+L}$.\\
{\bf Transmission: } & $(y^N,z^N)=\W^N x^N$\\
{\bf Inner decoding: } & Bob uses the standard decoder \cite{arikan09,honda12} with inputs $C_{(i/L)+1}$ and $y_{i+1}^{i+L}$ to obtain $\hat v_{i+1}^{i+L}$,\\
 & and hence $\hat t_{(i/L)+1}=\hat{v}_{i+1}^{i+L}[\mathcal E_K]$, for each $i\in \{0,L,\ldots,L(M-1) \}$.\\ 
{\bf Outer decoding: } & Bob computes $\hat u^{M}=\widetilde G_{M}^K \hat t^{M}$ and outputs a guess for the sent message $\hat m^J=\hat u^{M}[\mathcal{F}_J]$.\\
\vspace{-10mm}  
\end{tabular}
\begin{tabular}{c}
\hspace{51mm} \phantom{ {\bf{Protocol 1:}} Private channel coding}\hspace{53mm} \\ \vspace{-2.5mm} \\ \Xhline{3\arrayrulewidth}
\end{tabular}
\end{table}
%\end{savenotes}
\addtocounter{footnote}{-2}
\footnotetext{By the code construction the channel input distribution $P_X$ is defined. $P_X$ should be chosen such that it maximizes the scheme's rate.}
\footnotetext{Again an abuse of notation. See the Footnote~\ref{Footnote:Protocol1} of Protocol~\hyperlink{protocol:ska}{1}.}

\begin{figure}[!htb]
\centering
%\tikzsetnextfilename{Extractor}
\input{pcc.tex}
\caption{\small The private channel coding scheme for the setup $N=8$, $L=4$, $M=2$, $K=2$, and $J=2$. The message $M^J$ is first sent through an outer encoder which adds some bits (simulated as explained in \cite[Section IV]{honda12}) and applies the polarization transform $\widetilde G_{M}^K$. The output $T^{M}=(T_{(1)},\dots,T_{(K)})^M$ is then encoded a second time by $M$ independent identical blocks. Note that each block again adds redundancy (as explained in \cite[Definition 1]{sutter12}) before applying the polarization transform $G_L$. Each inner encoding block sends the frozen bits over a public channel to Bob. Note that this extra public communcation can be avoided as justified in Section~\ref{sec:ccpcc}. The output $X^N$ is then sent over $N$ copies of the wiretap channel $\W$ to Bob. Bob then applies a decoding operation as in the key agreement scheme, Section~\ref{sec:ska}.}
\label{fig:pcc}
\end{figure}
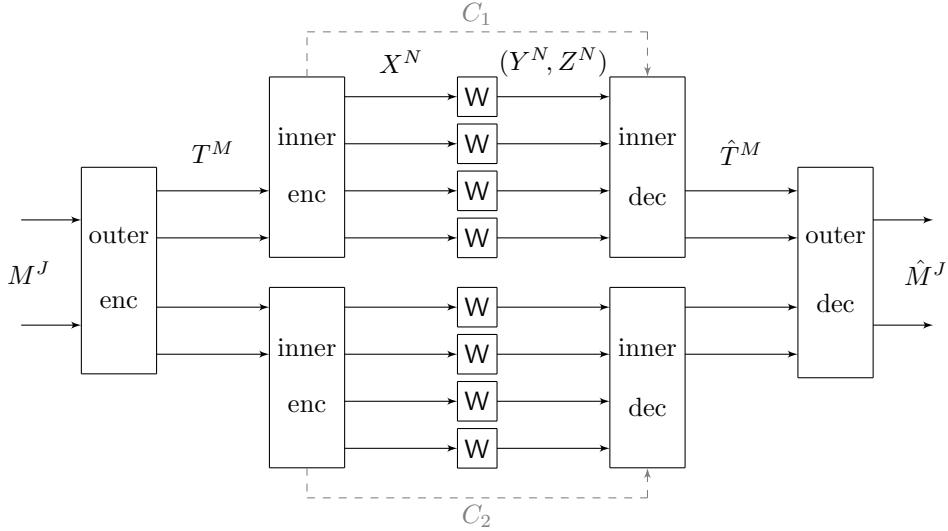

%%%%%%%%%%%%%%%%%%%%%%%%%%%

\subsection{Rate, Reliability, Secrecy, and Efficiency}
\begin{mycor}\label{cor:Sec}
For any $\beta < \tfrac{1}{2}$, Protocol~\hyperlink{protocol:pcc}{2} satisfies
\begin{alignat}{2}
&\textnormal{Reliability:}\quad&&\Prv{M^J \ne \hat M^J}= O\!\left(M 2^{-L^{\beta}}\right) \\
&\textnormal{Secrecy:}&&\norm{P_{M^J,Z^N,C}-\overline P_{M^J}\times P_{Z^N,C}}_1= O\!\left(\sqrt{N} 2^{-\frac{N^{\beta}}{2}}\right) \\
&\textnormal{Rate:}&&R=\Hc{X}{Z}-\frac{1}{L}\Hc{V^L[\mathcal{E}_K^\setC]}{Z^L}-\frac{o(N)}{N} \label{eq:ratePCC}
\end{alignat}
and its computational complexity is $O(N \log N)$.
\end{mycor}
\begin{proof}
Recall that the idea of the private channel coding scheme is to run Protocol~\hyperlink{protocol:ska}{1} backwards. Since Protocol~\hyperlink{protocol:pcc}{2} simulates the nearly deterministic bits $U^M[\mathcal{F}_J]$ at the outer encoder as described in \cite[Definition 1]{sutter12} and the almost random bits $V^L[\mathcal{E}_K^{\setC}]$ at the inner encoder as explained in \cite[Section IV]{honda12}, it follows that for large values of $L$ and $M$ the private channel coding scheme approximates the one-way secret-key scheme setup,\footnote{This approximation can be made arbitrarily precise for sufficiently large values of $L$ and $M$.} i.e.,\ $\lim_{N\to \infty}\TD{P_{T^M}}{P_{(V^L[\mathcal{E}_K])^M}}=0$ and $\lim_{L\to \infty}\TD{P_{X^L}}{P_{\hat X^L}}=0$ and, where $P_{X^L}$ denotes the distribution of the vector $X^L$ which is sent over the wiretap channel $\W$ and $P_{\hat X^L}$ denotes the distribution of Alice's random variable $\hat X^L$ in the one-way secret-key agreement setup. We thus can use the decoder introduced in \cite{arikan10} to decode the inner layer. Since we are using $M$ identical independent inner decoding blocks, by the union bound we obtain the desired reliability condition. The secrecy and rate statement are immediate consequences from Theorem~\ref{thm:main}.

\end{proof}

As mentioned after Theorem~\ref{thm:main}, to ensure reliability of the protocol, $M$ may not grow exponentially fast in $L$. 

\begin{mycor} \label{cor:ratePCC}
The rate of Protocol~\hyperlink{protocol:pcc}{2} given in Corollary~\ref{cor:Sec} can be bounded as
\begin{equation}
R \geq \max \left \lbrace 0, \Hc{X}{Z}-\Hc{X}{Y} - \frac{o(N)}{N} \right \rbrace.
\end{equation}
\end{mycor}
\begin{proof}
The proof is identical to the proof of Corollary~\ref{cor:rate}.
\end{proof}

\subsection{Achieving the Secrecy Capacity}
Corollaries~\ref{cor:PCCmoreCapable} and \ref{cor:ratePCC} immediately imply that our private channel coding scheme achieves the secrecy capacity for the setup where $\W$ is more capable.
If we can find the optimal auxiliary random variable $V$ in \eqref{eq:secrecyCapacity}, Protocol~\hyperlink{protocol:pcc}{2} can achieve the secrecy capacity for a general wiretap channel scenario. We define a super-channel $\overline{\W}:\mathcal{V}\to \mathcal{Y}\times \mathcal{Z}$ which includes the random variable $X$ and the wiretap channel $\W$. The super-channel $\overline{\W}$ is characterized by its transition probability distribution $P_{Y,Z|V}$ where $V$ is the optimal random variable solving  \eqref{eq:secrecyCapacity}. The private channel coding scheme is then applied to the super-channel, achieving the secrecy capacity. Note that finding the optimal random variable $V$ might be difficult.

In Section~\ref{sec:discussion}, we discuss the question if it is possible that Protocol~\hyperlink{protocol:pcc}{2} achieves a rate that is strictly larger than $\max \left \lbrace 0, \Hc{X}{Z}-\Hc{X}{Y} \right \rbrace$, if nothing about the optimal auxiliary random variable $V$ is known.

\subsection{Code Construction \& Public Channel Communication}
\label{sec:ccpcc}

To start the private channel coding scheme the code construction has to be done. Therefore, the index sets $\mathcal{E}_K$ and $\mathcal{F}_J$ as defined in \eqref{eq:Eset} and \eqref{eq:Fset} need to be computed. This can be done as explained in  Section~\ref{subsec:SKA_codeConstr}. The code construction defines the input distribution $P_X$ to the wiretap channel, which should be chosen such that it maximizes the scheme's rate given in \eqref{eq:ratePCC}.

We next explain how the communication $C^M \in \mathcal{C}^M$ from Alice to Bob can be reduced such that it does not affect the rate, i.e.,\ we show that we can choose $\left|\mathcal{C}\right|=o(L)$. Recall that we defined the index set $\mathcal{E}_K:=\mathcal{D}_{\epsilon_1}^L(X|Y)$ in \eqref{eq:Eset}. Let $\mathcal{G}:=\mathcal{R}_{\epsilon_1}^{L}(X|Y)$ using the noation introduced in \eqref{eq:Rset} and $\mathcal{I}:=[L]\backslash (\mathcal{E}_K \cup \mathcal{G})=\mathcal{E}_K^{\setC}\backslash \mathcal{G}$. As explained in Section~\ref{sec:polarization}, $\mathcal{G}$ consists of the outputs $V_j$ which are essentially uniformly random, even given all previous outputs $V^{j-1}$ as well as $Y^L$, where $V^L =G_L X^L$. The index set $\mathcal{I}$ consists of the outputs $V_j$ which are neither essentially uniformly random nor essentially deterministic given $V^{j-1}$ and $Y^L$. The polarization phenomenon stated in Theorem~\ref{thm:polPheno} ensures that this set is small, i.e.,\ that $\left|\mathcal{I}\right|=o(L)$. Since the bits of $\mathcal{G}$ are almost uniformly distributed, we can fix these bits independently of the message---as part of the code construction---without affecting the reliability of the scheme for large blocklengths.\footnote{Recall that we choose $\epsilon_1=O\left(2^{-L^{\beta}} \right)$ for $\beta<\tfrac{1}{2}$, such that for $L\to \infty$ the index set $\mathcal{G}$ contains only uniformly distributed bits.} We thus only need to communicate the bits belonging to the index set $\mathcal{I}$. 

We can send the bits belonging to $\mathcal{I}$ over a seperate public noiseless channel. Alternatively, we could send them over the wiretap channel $\W$ that we are using for private channel coding. However since $\W$ is assumed to be noisy and it is essential that the bits in $\mathcal{I}$ are recieved by Bob without any errors, we need to protect them using an error correcting code. To not destroy the essentially linear computational complexity of our scheme, the code needs to have an encoder and decoder that are practically efficient. Since $\left|\mathcal{I}\right|=o(L)$, we can use any error correcting code that has a non-vanishing rate. 
For symmetric binary DMCs, polar coding can be used to transmit reliably an arbitrarily distributed message \cite[Proposition 3]{vardy11}. We can therefore symmetrize our wiretap channel $\W$ and use polar codes to transmit the bits in $\mathcal{I}$.\footnote{Note that the symmetrization of the channel will reduce its rate which however does not matter as we need a non-vanishing rate only.}

%One family of codes suitable for this purpose would be expander codes \cite{sipser96}.\footnote{The bits belonging to $\mathcal{I}$ are not uniformly distributed and therefore polar codes might not be appropriate to use.}

As the reliability of the scheme is the average over the possible assignments of the random bits belonging to $\mathcal{I}$ (or even $\mathcal{E}_K^{\setC}$), at least one choice must be as good as the average, meaning a reliable, efficient, and deterministic scheme must exist. However, it might be computationally hard to find this choice.

% In the private channel coding scheme we introduce in this section, Alice sends some bits, denoted by $C_N$, over a public noiseless channel to Bob. These bits\footnote{They are often called \emph{frozen bits}.} are chosen using a random construction. As the reliability of the scheme is the average over the possible assignments of these random values, at least one choice must be as good as the average, meaning a reliable, efficient, and deterministic scheme must exist. Thinking of the random choices as part of the code construction, it follows by the Markov inequality that most choices will lead to schemes with these properties. We thus do not need the extra communication $C_N$. 
% 
% Nevertheless, we introduce the private channel coding scheme such that it uses the communcation $C_N$, since this leads to a stronger result in terms of security, in the sense that the scheme is secure even if we allow public transmission $C_N$ from Alice to Bob.
% In the following we are using the same notation as in Section~\ref{sec:ska}.

%%%%%%%%%%%%%%%%%%%%%%%%%%%%%%%%%%%
%%%%%%%%%%%%%%%%%%%%%%%%%%%%%%%%%%%
\section{Discussion} \label{sec:discussion} 
In this section, we describe two open problems, both of which address the question of whether rates beyond $\max \left \lbrace 0,\Hc{X}{Z}-\Hc{X}{Y}\right \rbrace$ can be achieved by our key agreement scheme, even if the optimal auxiliary random variables $V$ and $U$ are not given, i.e.,\ if we run Protocol~\hyperlink{protocol:ska}{1} directly for $X$ (instead of $U$) without making $V$ public. It may be even possible that the key agreement scheme achieves the optimal rate; no result to our knowledge implies otherwise. The two questions could also be formulated in the private coding scenario, whether rates beyond $\max \left \lbrace 0,\max_{P_X}\Hc{X}{Z}-\Hc{X}{Y}\right \rbrace$ are possible, but as positive answers in the former context imply positive answers in the latter, we shall restrict attention to the key agreement scenario for simplicity. 
% However, as such statements would follow from the key agreement case, we restrict attention to the key agreement 
%our schemes are optimal---without knowing anything about the auxiliary random variables. 

%%%%%%%%%%%%%%%%%%%%%%%%%%%%%%%%%%%

\subsection{Polarization with Bob's or Eve's Side Information}
%The first open question is related to the structure of the high entropy sets of $X^L$ with side information $Y^L$ or $Z^L$.
\renewcommand\themyquestion{1}
\begin{myquestion} \hypertarget{question}{}
Does for some distributions $P_{X,Y,Z}$ the rate of Protocol~\hyperlink{protocol:ska}{1} satisfy 
%For the two schemes introduced in Sections~\ref{sec:ska} and \ref{sec:pcc}, is it possible that 
\begin{equation}
R>\max\left \lbrace 0,\Hc{X}{Z}-\Hc{X}{Y}\right \rbrace, \quad  \textnormal{for} \quad N \to \infty? 
\end{equation}
\end{myquestion}
An equivalent formulation of this question is whether inequality $\eqref{eq:tight}$ is always tight for large enough $N$, i.e.,\
%Recall that we use the setup and notation introduced in Section~\ref{sec:ska}. According to $\eqref{eq:tight}$, another equivalent question is the following: 
\renewcommand\themyquestion{1'}
\begin{myquestion} \hypertarget{question1_prime}{}
 Is it possible that
\begin{equation}
\lim \limits_{L \to \infty} \frac{1}{L} \Hc{V^L[\mathcal{E}_K^{\setC}]}{Z^L} < \lim_{L\to \infty} \frac{1}{L}\left|\mathcal{E}_K^{\setC}  \right|, \quad \textnormal{for} \quad R > 0? \label{eq:question}
\end{equation}
\end{myquestion}
Using the polarization phenomenon stated in Theorem~\ref{thm:polPheno} we obtain 
\begin{equation}
\lim \limits_{L\to \infty} \frac{1}{L} \left|\mathcal{E}_K^{\setC}  \right| = \Hc{X}{Y},
\end{equation}
which together with \eqref{eq:question} would imply that $R>\max\left \lbrace 0,\Hc{X}{Z}-\Hc{X}{Y}\right \rbrace$ for $N \to \infty$ is possible. Relation \eqref{eq:question} can only be satisfied if the high-entropy set with respect to Bob's side information, i.e.,\ the set $\mathcal{E}^{\setC}_K$, is not always a high-entropy set with respect to Eve's side information. Thus, the question of rates in the key agreement protocol is closely related to fundamental structural properties of the polarization phenomenon. 

For less noisy channels $\W$ defined by $P_{YZ|X}$ (cf.\ Section~\ref{sec:defn}), these questions can be answered in the negative.  In this case we have $H(X^L|Z^L)\geq H(X^L|Y^L)$, and since $V^L[\mathcal E_K^\setC]$ is a deterministic function of $X^L$,%From \eqref{eq:question} we have  
\begin{equation}
\lim \limits_{L\to \infty} \frac{1}{L} \Hc{V^L[\mathcal{E}_K^{\setC}]}{Z^L} \geq  \lim \limits_{L \to \infty} \frac{1}{L}\Hc{V^L[\mathcal{E}_K^{\setC}]}{Y^L} = \lim \limits_{L \to \infty} \frac{1}{L} \left|\mathcal{E}_K^{\setC}  \right|. \label{eq:equalitY}
\end{equation}
Thus, \eqref{eq:question} cannot hold. The final equality can be justified as follows. Recall that we defined  $\mathcal{E}_K:=\mathcal{D}_{\epsilon_1}^L(X|Y)$ in \eqref{eq:Eset}. Let $\mathcal{H}_{L-K}:=\mathcal{R}_{\epsilon_1}^L(X|Y)$ and $\mathcal{I}:=[L]\backslash (\mathcal{E}_K \cup \mathcal{H}_{L-K})$ such that $\mathcal{E}_K^{\setC}=\mathcal{H}_{L-K}\cup \mathcal{I}$. Recall that we can choose $\epsilon_1=O(2^{-L^{\beta_1}})$ for $\beta_1 <\frac{1}{2}$. Using the chain rule and the polarization phenomenon given in Theorem~\ref{thm:polPheno}, we obtain
\begin{align}
  \lim \limits_{L \to \infty} \frac{1}{L}\Hc{V^L[\mathcal{E}_K^{\setC}]}{Y^L} &=  \lim \limits_{L \to \infty} \frac{1}{L} \sum_{i \in \mathcal{E}_K}\Hc{V^L[\mathcal{E}_K^{\setC}]_i}{V^L[\mathcal{E}_K^{\setC}]^{i-1},Y^L}\\
 &\geq \lim \limits_{L \to \infty} \frac{1}{L} \left( \left(1-\epsilon_1\right)\left|\mathcal{H}_{L-K}\right|+\epsilon_1 \left|\mathcal{I} \right|\right)\\
 &= \lim \limits_{L \to \infty} \frac{1}{L} \left| \mathcal{E}_{K}^{\setC} \right|.
\end{align}
Using the upper bound of the entropy in terms of the alphabet size we conclude that the equality in \eqref{eq:equalitY} holds.
The fact that \eqref{eq:question} is not possible in the setup where $\W$ is less noisy accords with the one-way secret-key rate formula given in \eqref{eq:skaLessNoisy}, which excludes rates beyond $\max\left \lbrace 0, \Hc{X}{Z}-\Hc{X}{Y}\right \rbrace$.

If the answer to Question~\hyperlink{question}{1}, or equivalently to Question~\hyperlink{question1_prime}{1'}, is ``yes'', this would give some new insights into the problem of finding the optimal auxiliary random variables $U,V$ in \eqref{eq:highestRate} (and $V$ in \eqref{eq:secrecyCapacity}), which may be hard in general. 

Furthermore, a positive answer to Question~\hyperlink{question}{1} implies that we can send quantum information reliable over a quantum channel at a rate that is beyond the \emph{coherent information} using the scheme introduced in \cite{concEntagDist}. Since the best known achievable rate for a wide class of quantum channels is the coherent information, our scheme would improve this bound. Furthermore, it would be of interest to know by how much we can outperform the coherent information.\footnote{Since there exist a lot of good converse bounds for sending quantum information reliable over an arbitrary quantum channel \cite{bennett_mixed-state_1996,smith08,smith08_2}, it would be interesting to see how closely they can be met.}

%%%%%%%%%%%%%%%%%%%%%%%%%%%%%%%%%%%
\subsection{Approximately Less Noisy Super-Channel} \label{sec:diss_LN}
To state the second open problem,  consider the super-source which outputs the triple of random variables  $(V^L[\mathcal{E}_K],(Y^L, V^L[\mathcal{E}_K^{\setC}]),(Z^L, V^L[\mathcal{E}_K^{\setC}]))$. For instance, Figure~\ref{fig:SKA} consists of two super-sources. The super-source implicitly defines a super-channel $\W'$ using the conditional probability distribution of the second two random variables given the first. Then we have

\begin{myprop} \label{prop:superChannel_LN}
For sufficiently large $L$, the channel $\W'$ is approximately less noisy, irrespective of $\W$.
\end{myprop}
\begin{proof}
Using the chain rule we can write
\begin{align}
 \Hc{V^L[\mathcal{E}_K]}{V^L[\mathcal{E}_K^{\setC}],Y^L} &= \sum_{i \in \mathcal{E}_K} \Hc{V^L[\mathcal{E}_K]_i}{V^L[\mathcal{E}_K]^{i-1},V^L[\mathcal{E}_K^{\setC}],Y^L}\\
 &\leq \sum_{i \in \mathcal{E}_K}\Hc{V_i}{V^{i-1},Y^L}\\
 &\leq K \epsilon_1,
\end{align}
where the last inquality follows by definition of the set $\mathcal{E}_K$. Recall that we can choose $\epsilon_1=O\left(2^{-L^{\beta}} \right)$ for $\beta <\tfrac{1}{2}$. The polarization phenomenon stated in Theorem~\ref{thm:polPheno} ensures that $K=O\left(L \right)$. Hence, we can apply the following Lemma~\ref{lem:LessNoisy} which proves the assertion.
\end{proof}

\begin{mylem} \label{lem:LessNoisy}
If $U \markovDavid X  \markovDavid (Y,Z)$ form a Markov chain in the given order and $\Hc{X}{Y}\leq \epsilon$ for $\epsilon\geq 0$, then $\Hc{U}{Y}\leq \Hc{U}{Z}+\epsilon$ for all possible distributions of $(U,X)$.
\end{mylem}
\begin{proof}
Using the chain rule and the non-negativity of the entropy we can write
\begin{align}
 \Hc{U}{Y}&\leq \Hc{U}{Y}+\Hc{X}{Y,U}\\
&=\Hc{U,X}{Y}\\
&=\Hc{X}{Y}+\Hc{U}{X,Y}\\
&\leq \epsilon + \Hc{U}{X} \label{eq:condred}\\
&\leq \epsilon +\Hc{U}{Z}. \label{eq:dataproc}
\end{align}
Inequality~\eqref{eq:condred} follows by assumption and since conditioning reduces entropy. The final inequality uses the data processing inequality.
\end{proof}

Proposition~\ref{prop:superChannel_LN} and Lemma~\ref{lem:LessNoisy} imply that the DM-WTC $\W'$ induced by the super-source described above is almost less noisy. More precisely we have for $\beta<\tfrac{1}{2}$ and  $\xi=O\left(L 2^{-L^{\beta}} \right)$
\begin{equation}
\Hc{T}{V^L[\mathcal{E}_K^{\setC}],Y^L} \leq \Hc{T}{V^L[\mathcal{E}_K^{\setC}],Z^L} + \xi, \label{eq:epsilonLN}
\end{equation}
for all possible distributions of $T$, where $T \markovDavid V^L[\mathcal{E_K}] \markovDavid ((Y^L, V^L[\mathcal{E}_K^{\setC}]),(Z^L, V^L[\mathcal{E}_K^{\setC}]))$ and $\left|\mathcal{T}\right| \leq K$. Following the proof of Corollary~\ref{cor:LN_SKA}---using \eqref{eq:epsilonLN} in \eqref{eq:IneqLN}---we obtain the one-way secret-key rate of the super-source as
\begin{align}
&\frac{1}{L}\SK{V^L[\mathcal{E}_K]}{Y^L,V^L[\mathcal{E}_K^{\setC}]}{Z^L,V^L[\mathcal{E}_K^{\setC}]}\nonumber \\
&\hspace{30mm}= \frac{1}{L}\left(\Hc{V^L[\mathcal{E}_K]}{Z^L,V^L[\mathcal{E}_K^{\setC}]}-\Hc{V^L[\mathcal{E}_K]}{Y^L,V^L[\mathcal{E}_K^{\setC}]} +\xi \right)\\
&\hspace{30mm}=\frac{1}{L}\left(\Hc{V^L[\mathcal{E}_K]}{Z^L,V^L[\mathcal{E}_K^{\setC}]} \right) - \frac{o(N)}{N}\\
&\hspace{30mm}=R. \label{eq:rateSuper}
\end{align}
The second equation follows by definition of the set $\mathcal{E}_K$ and \eqref{eq:rateSuper} is according to \eqref{eq:rateForLater}. We thus conclude that the one-way secret-key agreement scheme introduced in Section~\ref{sec:ska} always achieves the one-way secret-key rate for the super-source as defined above.  %(Similarly, the private coding scheme for $\W'$ achieves the secrecy capacity of $\W'$ irrespective of $\W$.) 
This raises the question of when the super-source has the same key rate as the original source, i.e.,\ how much is is lost in the first layer of our key agreement scheme. 
\renewcommand\themyquestion{2}
\begin{myquestion}  \hypertarget{question:SKA_SuperSource}{}
 For what conditions does $\tfrac{1}{L}\SK{V^L[\mathcal{E}_K]}{Y^L,V^L[\mathcal{E}_K^{\setC}]}{Z^L,V^L[\mathcal{E}_K^{\setC}]}=\SK{X}{Y}{Z}$ hold? 
\end{myquestion}
Having $\tfrac{1}{L}\SK{V^L[\mathcal{E}_K]}{Y^L,V^L[\mathcal{E}_K^{\setC}]}{Z^L,V^L[\mathcal{E}_K^{\setC}]}=\SK{X}{Y}{Z}$ implies that Protocol~\hyperlink{protocol:ska}{1} achieves the one-way secret-key rate without knowing anything about the optimal auxiliary random variables $V$ and $U$. If $\W$ is less noisy, Corollary~\ref{cor:LN_SKA} ensures that $\tfrac{1}{L}\SK{V^L[\mathcal{E}_K]}{Y^L,V^L[\mathcal{E}_K^{\setC}]}{Z^L,V^L[\mathcal{E}_K^{\setC}]}=\SK{X}{Y}{Z}$ must be satisfied. For other scenarios Question~\hyperlink{question:SKA_SuperSource}{2} is currently unsolved.   

For the setup of private channel coding, following the proof of Corollary~\ref{cor:PCCmoreCapable} using \eqref{eq:epsilonLN} shows that the secrecy capacity of the super-channel $\W'$ is
\begin{align}
C_s(\W')&= \frac{1}{L}\left(\Hc{V^L[\mathcal{E}_K]}{Z^L,V^L[\mathcal{E}_K^{\setC}]}-\Hc{V^L[\mathcal{E}_K]}{Y^L,V^L[\mathcal{E}_K^{\setC}]} + \xi \right)\\
&=\frac{1}{L}\left(\Hc{V^L[\mathcal{E}_K]}{Z^L,V^L[\mathcal{E}_K^{\setC}]} \right) - \frac{o(N)}{N}\\
&=R. \label{eq:rateSuperChannel}
\end{align}
The scheme introduced in Protocol~\hyperlink{protocol:pcc}{2} hence achieves the secrecy capacity for the channel $\W'$ irrespective of the channel $\W$. This raises the question when the super-channel and the original channel have the same secrecy capacity.
\renewcommand\themyquestion{2'}
\begin{myquestion}  \hypertarget{question:pcc}{}
 Under what conditions does $C_s(\W')=C_s(\W)$ hold?
\end{myquestion}
$C_s(\W')=C_s(\W)$ being valid implies that Protocol~\hyperlink{protocol:pcc}{2} achieves the secrecy capacity of $\W$ without having knowledge about the optimal auxiliary random variable $V$. If $\W$ is more capable, according to Corollary~\ref{cor:PCCmoreCapable} $C_s(\W')=C_s(\W)$ must hold. For other channels, Question~\hyperlink{question:pcc}{2'} has not yet been resolved. 

%%%%%%%%%%%%%%%%%%%%%%%%%%%%%%%%%%%

\subsection{Conclusion}
We have constructed practically efficient protocols (with complexity essentially linear in the blocklength) for one-way secret-key agreement from correlated randomness and for private channel coding over discrete memoryless wiretap channels. Each protocol achieves the corresponding optimal rate. Compared to previous methods, we do not require any degradability assumptions and achieve strong (rather than weak) secrecy.

Our scheme is formulated for arbitrary discrete memoryless wiretap channels. Using ideas of \c Sa\c so\u glu et al.\ \cite{sasoglu09} the two protocols presented in this paper can also be used for wiretap channels with continuous input alphabets.

%%%%%%%%%%%%%%%%%%%%%%%%%%%%%%%%%%
%%%%%%%%%%%%%%%%%%%%%%%%%%%%%%%%%%%

 \section*{Acknowledgments}
The authors would like to thank Alexander Vardy for useful discussions.
This work was supported by the Swiss National Science Foundation (through the National Centre of Competence in Research `Quantum Science and Technology' and grant No.~200020-135048) and by the European Research Council (grant No.~258932).

%%%%%%%%%%%%%%%%%%%%%%%%%%%%%%%%%%
%%%%%%%%%%%%%%%%%%%%%%%%%%%%%%%%%%%

% Using ieeetr. the ack is at the title page

% \section*{Acknowledgments}
% This work was supported by the Swiss National Science Foundation (through the National Centre of Competence in Research `Quantum Science and Technology' and grant No.~200020-135048) and by the European Research Council (grant No.~258932).

%\bibliographystyle{ieeetr} % use here the style you prefer :D ieeetr for TIT, apsrev4-1 for APS    TO BE ERASED AGAIN WHEN I HAVE MY NICE MACBOOK BACK!!!!!!!!
%\bibliographystyle{apsrev4-1}
\bibliography{./bibtex/header,./bibtex/bibliofile}

%%%%%%%%%%%%%%%%%
%%%%%%%%%%%%%%%%%
%%%%%%%%%%%%%%%%%
%%%%%%%%%%%%%%%%%

%\clearpage

%%%%%%%%%%%%%%%%%
%%%%%%%%%%%%%%%%%
%%%%%%%%%%%%%%%%%

%\widetext

\end{document}

%% file: SKAscheme.tex
\def \xcomp{0.9}
\def \ycomp{5}
\def \xblock{2.93}
\def \yblock{3}

\def \yPA{1.75}

\def \xgapd{0.9} %gap between decompressors
\def \ygapd{0.9} %gap between decompressors

\def \ybs{-0.3} %height block 

\def \ys{0.1} % in dashed x-channel

\def \xdec{0.9} % xvalue decompressor
\def \w{0.526} % small W-block

\def \xs{0.4} % arrows at the border

\def \s{2.8} % shift for second PA block

\def \c{0.4} % frozen input length

\def \a{0.1} % antenna

\def \la{0.2} % for labeling

\def \xdash{0.2} %dist. dashed boxes
\def \ydash{0.2} %dist. dashed boxes

\begin{tikzpicture}[scale=1,auto, node distance=1cm,>=latex']
	
    % PA
     \draw [draw] (\xgapd+\xcomp,-2.5*\w-2*\ys+0.5*\w+0.5*\ys) -- (2*\xcomp+\xgapd,-2.5*\w-2*\ys+0.5*\w+0.5*\ys); 
     \draw [draw] (\xgapd+\xcomp,-2.5*\w-2*\ys+0.5*\w+0.5*\ys) -- (\xgapd+\xcomp,-\yPA-2*\ybs-\s);
     \draw [draw] (\xgapd+\xcomp,-\yPA-2*\ybs-\s) -- (2*\xcomp+\xgapd,-\yPA-2*\ybs-\s);
     \draw [draw] (2*\xcomp+\xgapd,-2.5*\w-2*\ys+0.5*\w+0.5*\ys) -- (2*\xcomp+\xgapd,-\yPA-2*\ybs-\s);     
    \node at (1.5*\xcomp+\xgapd,-0.66*\yPA-1.32*\ybs-0.66*\s+2*\la) {PA};
    \node at (1.5*\xcomp+\xgapd,-0.66*\yPA-1.32*\ybs-0.66*\s-2*\la) {$\widetilde G_{M}^K$};

 % IR FIRST
  \draw [draw] (2*\xcomp+3*\xgapd,0) -- (2*\xcomp+3*\xgapd+\xcomp,0);  
  \draw [draw] (2*\xcomp+3*\xgapd,-\yblock-2*\ybs) -- (2*\xcomp+3*\xgapd+\xcomp,-\yblock-2*\ybs);
  \draw [draw] (2*\xcomp+3*\xgapd,0) -- (2*\xcomp+3*\xgapd,-\yblock-2*\ybs);  
 \draw [draw] (2*\xcomp+3*\xgapd+\xcomp,-\yblock-2*\ybs) -- (2*\xcomp+3*\xgapd+\xcomp,0);
     \node [] at (2*\xcomp+3*\xgapd+0.5*\xcomp,-0.33*\yblock-0.66*\ybs) {IR};
     \node [] at (2*\xcomp+3*\xgapd+0.5*\xcomp,-0.66*\yblock-1.32*\ybs) {$G_L$};
     
    \draw [draw] (3*\xcomp+4*\xgapd,0) -- (3*\xcomp+4*\xgapd+\w,0);  
    \draw [draw] (3*\xcomp+4*\xgapd,-\w) -- (3*\xcomp+4*\xgapd+\w,-\w);
    \draw [draw] (3*\xcomp+4*\xgapd,0) -- (3*\xcomp+4*\xgapd,-\w);
    \draw [draw] (3*\xcomp+4*\xgapd+\w,0) -- (3*\xcomp+4*\xgapd+\w,-\w);
   %\node at (3*\xcomp+4*\xgapd+0.5*\w,-0.5*\w) {$\W$};
    \node[draw,cloud,cloud puffs = 10,pattern=north east lines,minimum width=0.4cm,minimum height=0.3cm] at (3*\xcomp+4*\xgapd+0.5*\w,-0.5*\w){};
   
    \draw [draw] (3*\xcomp+4*\xgapd,-\ys-\w) -- (3*\xcomp+4*\xgapd+\w,-\ys-\w);  
    \draw [draw] (3*\xcomp+4*\xgapd,-\w-\ys-\w) -- (3*\xcomp+4*\xgapd+\w,-2*\w-\ys);
    \draw [draw] (3*\xcomp+4*\xgapd,-\ys-\w) -- (3*\xcomp+4*\xgapd,-2*\w-\ys);
    \draw [draw] (3*\xcomp+4*\xgapd+\w,-\ys-\w) -- (3*\xcomp+4*\xgapd+\w,-2*\w-\ys);
   %\node at (3*\xcomp+4*\xgapd+0.5*\w,-1.5*\w-\ys) {$\W$};
    \node[draw,cloud,cloud puffs = 10,pattern=north east lines,minimum width=0.4cm,minimum height=0.3cm] at (3*\xcomp+4*\xgapd+0.5*\w,-1.5*\w-\ys){};   
   
    \draw [draw] (3*\xcomp+4*\xgapd,-2*\ys-2*\w) -- (3*\xcomp+4*\xgapd+\w,-2*\ys-2*\w);  
    \draw [draw] (3*\xcomp+4*\xgapd,-3*\w-2*\ys) -- (3*\xcomp+4*\xgapd+\w,-3*\w-2*\ys);
    \draw [draw] (3*\xcomp+4*\xgapd,-2*\ys-2*\w) -- (3*\xcomp+4*\xgapd,-3*\w-2*\ys);
    \draw [draw] (3*\xcomp+4*\xgapd+\w,-2*\ys-2*\w) -- (3*\xcomp+4*\xgapd+\w,-3*\w-2*\ys);
   %\node at (3*\xcomp+4*\xgapd+0.5*\w,-2.5*\w-2*\ys) {$\W$};
    \node[draw,cloud,cloud puffs = 10,pattern=north east lines,minimum width=0.4cm,minimum height=0.3cm] at (3*\xcomp+4*\xgapd+0.5*\w,-2.5*\w-2*\ys){};     
   
    \draw [draw] (3*\xcomp+4*\xgapd,-3*\ys-3*\w) -- (3*\xcomp+4*\xgapd+\w,-3*\ys-3*\w);  
    \draw [draw] (3*\xcomp+4*\xgapd,-4*\w-3*\ys) -- (3*\xcomp+4*\xgapd+\w,-4*\w-3*\ys);
    \draw [draw] (3*\xcomp+4*\xgapd,-3*\ys-3*\w) -- (3*\xcomp+4*\xgapd,-4*\w-3*\ys);
    \draw [draw] (3*\xcomp+4*\xgapd+\w,-3*\ys-3*\w) -- (3*\xcomp+4*\xgapd+\w,-4*\w-3*\ys);
%   \node at (3*\xcomp+4*\xgapd+0.5*\w,-3.5*\w-3*\ys) {$\W$};
    \node[draw,cloud,cloud puffs = 10,pattern=north east lines,minimum width=0.4cm,minimum height=0.3cm] at (3*\xcomp+4*\xgapd+0.5*\w,-3.5*\w-3*\ys){}; 

 % IR SECOND
  \draw [draw] (2*\xcomp+3*\xgapd,-\s) -- (2*\xcomp+3*\xgapd+\xcomp,-\s);  
  \draw [draw] (2*\xcomp+3*\xgapd,-\yblock-2*\ybs-\s) -- (2*\xcomp+3*\xgapd+\xcomp,-\yblock-2*\ybs-\s);
  \draw [draw] (2*\xcomp+3*\xgapd,-\s) -- (2*\xcomp+3*\xgapd,-\yblock-2*\ybs-\s);  
 \draw [draw] (2*\xcomp+3*\xgapd+\xcomp,-\yblock-2*\ybs-\s) -- (2*\xcomp+3*\xgapd+\xcomp,-\s);
     \node [] at (2*\xcomp+3*\xgapd+0.5*\xcomp,-0.33*\yblock-0.66*\ybs-\s) {IR};
     \node [] at (2*\xcomp+3*\xgapd+0.5*\xcomp,-0.66*\yblock-1.32*\ybs-\s) {$G_L$};
          
    \draw [draw] (3*\xcomp+4*\xgapd,-\s) -- (3*\xcomp+4*\xgapd+\w,-\s);  
    \draw [draw] (3*\xcomp+4*\xgapd,-\w-\s) -- (3*\xcomp+4*\xgapd+\w,-\w-\s);
    \draw [draw] (3*\xcomp+4*\xgapd,-\s) -- (3*\xcomp+4*\xgapd,-\w-\s);
    \draw [draw] (3*\xcomp+4*\xgapd+\w,-\s) -- (3*\xcomp+4*\xgapd+\w,-\w-\s);
%   \node at (3*\xcomp+4*\xgapd+0.5*\w,-0.5*\w-\s) {$\W$};
    \node[draw,cloud,cloud puffs = 10,pattern=north east lines,minimum width=0.4cm,minimum height=0.3cm] at (3*\xcomp+4*\xgapd+0.5*\w,-0.5*\w-\s){};    
   
    \draw [draw] (3*\xcomp+4*\xgapd,-\ys-\w-\s) -- (3*\xcomp+4*\xgapd+\w,-\ys-\w-\s);  
    \draw [draw] (3*\xcomp+4*\xgapd,-\w-\ys-\w-\s) -- (3*\xcomp+4*\xgapd+\w,-2*\w-\ys-\s);
    \draw [draw] (3*\xcomp+4*\xgapd,-\ys-\w-\s) -- (3*\xcomp+4*\xgapd,-2*\w-\ys-\s);
    \draw [draw] (3*\xcomp+4*\xgapd+\w,-\ys-\w-\s) -- (3*\xcomp+4*\xgapd+\w,-2*\w-\ys-\s);
%   \node at (3*\xcomp+4*\xgapd+0.5*\w,-1.5*\w-\ys-\s) {$\W$};
    \node[draw,cloud,cloud puffs = 10,pattern=north east lines,minimum width=0.4cm,minimum height=0.3cm] at (3*\xcomp+4*\xgapd+0.5*\w,-1.5*\w-\ys-\s){};       
   
    \draw [draw] (3*\xcomp+4*\xgapd,-2*\ys-2*\w-\s) -- (3*\xcomp+4*\xgapd+\w,-2*\ys-2*\w-\s);  
    \draw [draw] (3*\xcomp+4*\xgapd,-3*\w-2*\ys-\s) -- (3*\xcomp+4*\xgapd+\w,-3*\w-2*\ys-\s);
    \draw [draw] (3*\xcomp+4*\xgapd,-2*\ys-2*\w-\s) -- (3*\xcomp+4*\xgapd,-3*\w-2*\ys-\s);
    \draw [draw] (3*\xcomp+4*\xgapd+\w,-2*\ys-2*\w-\s) -- (3*\xcomp+4*\xgapd+\w,-3*\w-2*\ys-\s);
%   \node at (3*\xcomp+4*\xgapd+0.5*\w,-2.5*\w-2*\ys-\s) {$\W$};
    \node[draw,cloud,cloud puffs = 10,pattern=north east lines,minimum width=0.4cm,minimum height=0.3cm] at (3*\xcomp+4*\xgapd+0.5*\w,-2.5*\w-2*\ys-\s){};   
   
    \draw [draw] (3*\xcomp+4*\xgapd,-3*\ys-3*\w-\s) -- (3*\xcomp+4*\xgapd+\w,-3*\ys-3*\w-\s);  
    \draw [draw] (3*\xcomp+4*\xgapd,-4*\w-3*\ys-\s) -- (3*\xcomp+4*\xgapd+\w,-4*\w-3*\ys-\s);
    \draw [draw] (3*\xcomp+4*\xgapd,-3*\ys-3*\w-\s) -- (3*\xcomp+4*\xgapd,-4*\w-3*\ys-\s);
    \draw [draw] (3*\xcomp+4*\xgapd+\w,-3*\ys-3*\w-\s) -- (3*\xcomp+4*\xgapd+\w,-4*\w-3*\ys-\s);
  % \node at (3*\xcomp+4*\xgapd+0.5*\w,-3.5*\w-3*\ys-\s) {$\W$};
    \node[draw,cloud,cloud puffs = 10,pattern=north east lines,minimum width=0.4cm,minimum height=0.3cm] at (3*\xcomp+4*\xgapd+0.5*\w,-3.5*\w-3*\ys-\s){};   
    
% decoder at Bob's side
    \draw [draw] (3*\xcomp+5*\xgapd+\w,0) -- (3*\xcomp+5*\xgapd+\w+\xdec,0); 
    \draw [draw] (3*\xcomp+5*\xgapd+\w,-\yblock-2*\ybs) -- (3*\xcomp+5*\xgapd+\w+\xdec,-\yblock-2*\ybs); 
    \draw [draw] (3*\xcomp+5*\xgapd+\w,0) -- (3*\xcomp+5*\xgapd+\w,-\yblock-2*\ybs);  
    \draw [draw] (3*\xcomp+5*\xgapd+\w+\xdec,-\yblock-2*\ybs) -- (3*\xcomp+5*\xgapd+\w+\xdec,0);  
    \node at (3*\xcomp+5*\xgapd+\w+0.5*\xdec,-0.5*\yblock-1*\ybs) {$\dec$};         
    
    \draw [draw] (3*\xcomp+5*\xgapd+\w,-\s) -- (3*\xcomp+5*\xgapd+\w+\xdec,-\s); 
    \draw [draw] (3*\xcomp+5*\xgapd+\w,-\yblock-2*\ybs-\s) -- (3*\xcomp+5*\xgapd+\w+\xdec,-\yblock-2*\ybs-\s); 
    \draw [draw] (3*\xcomp+5*\xgapd+\w,-\s) -- (3*\xcomp+5*\xgapd+\w,-\yblock-2*\ybs-\s);  
    \draw [draw] (3*\xcomp+5*\xgapd+\w+\xdec,-\yblock-2*\ybs-\s) -- (3*\xcomp+5*\xgapd+\w+\xdec,-\s);   
    \node at (3*\xcomp+5*\xgapd+\w+0.5*\xdec,-0.5*\yblock-1*\ybs-\s) {$\dec$};         
    
% Bob's inverse PA block   
    \draw [draw] (3*\xcomp+7*\xgapd+\w+\xdec,-2.5*\w-2*\ys+0.5*\w+0.5*\ys) -- (3*\xcomp+7*\xgapd+\w+\xdec+\xcomp,-2.5*\w-2*\ys+0.5*\w+0.5*\ys);  
    \draw [draw] (3*\xcomp+7*\xgapd+\w+\xdec,-1.5*\w-\ys-\s-0.5*\w-0.5*\ys) -- (3*\xcomp+7*\xgapd+\w+\xdec+\xcomp,-1.5*\w-\ys-\s-0.5*\w-0.5*\ys);      
    \draw [draw] (3*\xcomp+7*\xgapd+\w+\xdec,-2.5*\w-2*\ys+0.5*\w+0.5*\ys) -- (3*\xcomp+7*\xgapd+\w+\xdec,-1.5*\w-\ys-\s-0.5*\w-0.5*\ys);   
    \draw [draw] (3*\xcomp+7*\xgapd+\w+\xdec+\xcomp,-1.5*\w-\ys-\s-0.5*\w-0.5*\ys) -- (3*\xcomp+7*\xgapd+\w+\xdec+\xcomp,-2.5*\w-2*\ys+0.5*\w+0.5*\ys);
    \node at (3*\xcomp+7*\xgapd+\w+\xdec+0.5*\xcomp,-1.5*\w-\ys-0.5*\s-0.5*\w-0.5*\ys) {$\widetilde G_{M}^K$};                  
% Arrows

    \draw [<-] (2*\xcomp+3*\xgapd-\c,-0.5*\w) -- (2*\xcomp+3*\xgapd,-0.5*\w);          
    \draw [<-] (2*\xcomp+3*\xgapd-\c,-1.5*\w-\ys) -- (2*\xcomp+3*\xgapd,-1.5*\w-\ys);             
    \draw [<-] (2*\xcomp+\xgapd,-2.5*\w-2*\ys) -- (2*\xcomp+3*\xgapd,-2.5*\w-2*\ys);
    \draw [<-] (2*\xcomp+\xgapd,-3.5*\w-3*\ys) -- (2*\xcomp+3*\xgapd,-3.5*\w-3*\ys);        
  
     % -----
         \draw [<-] (\xcomp+\xgapd-\xgapd,-2*\w-1.5*\ys-0.125*\s) -- (\xcomp+\xgapd,-2*\w-1.5*\ys-0.125*\s);  
         \draw [<-] (\xcomp+\xgapd-\xgapd,-2*\w-1.5*\ys-0.375*\s) -- (\xcomp+\xgapd,-2*\w-1.5*\ys-0.375*\s);    
         \draw [<-] (\xcomp+0.5*\xgapd,-2*\w-1.5*\ys-0.625*\s) -- (\xcomp+\xgapd,-2*\w-1.5*\ys-0.625*\s);               
         \draw [<-] (\xcomp+0.5*\xgapd,-2*\w-1.5*\ys-0.875*\s) -- (\xcomp+\xgapd,-2*\w-1.5*\ys-0.875*\s);  
       
  % ------      

    \draw [<-] (2*\xcomp+\xgapd,-0.5*\w-\s) -- (2*\xcomp+3*\xgapd,-0.5*\w-\s); 
    \draw [<-] (2*\xcomp+\xgapd,-1.5*\w-\ys-\s) -- (2*\xcomp+3*\xgapd,-1.5*\w-\ys-\s);
    \draw [<-] (2*\xcomp+3*\xgapd-\c,-2.5*\w-2*\ys-\s) -- (2*\xcomp+3*\xgapd,-2.5*\w-2*\ys-\s);         
    \draw [<-] (2*\xcomp+3*\xgapd-\c,-3.5*\w-3*\ys-\s) -- (2*\xcomp+3*\xgapd,-3.5*\w-3*\ys-\s);

  % ------   
    \draw [<-] (3*\xcomp+3*\xgapd,-0.5*\w) -- (3*\xcomp+4*\xgapd,-0.5*\w); 
    \draw [->] (3*\xcomp+4*\xgapd+\w,-0.5*\w) -- (3*\xcomp+5*\xgapd+\w,-0.5*\w); 
    
    \draw [<-] (3*\xcomp+3*\xgapd,-1.5*\w-\ys) -- (3*\xcomp+4*\xgapd,-1.5*\w-\ys); 
    \draw [->] (3*\xcomp+4*\xgapd +\w,-1.5*\w-\ys) -- (3*\xcomp+5*\xgapd+\w,-1.5*\w-\ys);
        
    \draw [<-] (3*\xcomp+3*\xgapd,-2.5*\w-2*\ys) -- (3*\xcomp+4*\xgapd,-2.5*\w-2*\ys); 
    \draw [->] (3*\xcomp+4*\xgapd+\w,-2.5*\w-2*\ys) -- (3*\xcomp+5*\xgapd+\w,-2.5*\w-2*\ys);
     
    \draw [<-] (3*\xcomp+3*\xgapd,-3.5*\w-3*\ys) -- (3*\xcomp+4*\xgapd,-3.5*\w-3*\ys);             
    \draw [->] (3*\xcomp+4*\xgapd+\w,-3.5*\w-3*\ys) -- (3*\xcomp+5*\xgapd+\w,-3.5*\w-3*\ys);
  % ------   
    \draw [<-] (3*\xcomp+3*\xgapd,-0.5*\w-\s) -- (3*\xcomp+4*\xgapd,-0.5*\w-\s); 
    \draw [->] (3*\xcomp+4*\xgapd+\w,-0.5*\w-\s) -- (3*\xcomp+5*\xgapd+\w,-0.5*\w-\s); 
    
    \draw [<-] (3*\xcomp+3*\xgapd,-1.5*\w-\ys-\s) -- (3*\xcomp+4*\xgapd,-1.5*\w-\ys-\s); 
    \draw [->] (3*\xcomp+4*\xgapd +\w,-1.5*\w-\ys-\s) -- (3*\xcomp+5*\xgapd+\w,-1.5*\w-\ys-\s);
        
    \draw [<-] (3*\xcomp+3*\xgapd,-2.5*\w-2*\ys-\s) -- (3*\xcomp+4*\xgapd,-2.5*\w-2*\ys-\s); 
    \draw [->] (3*\xcomp+4*\xgapd+\w,-2.5*\w-2*\ys-\s) -- (3*\xcomp+5*\xgapd+\w,-2.5*\w-2*\ys-\s);
     
    \draw [<-] (3*\xcomp+3*\xgapd,-3.5*\w-3*\ys-\s) -- (3*\xcomp+4*\xgapd,-3.5*\w-3*\ys-\s);             
    \draw [->] (3*\xcomp+4*\xgapd+\w,-3.5*\w-3*\ys-\s) -- (3*\xcomp+5*\xgapd+\w,-3.5*\w-3*\ys-\s);
% ------------

    \draw [->] (3*\xcomp+5*\xgapd+\w+\xdec,-2.5*\w-2*\ys) -- (3*\xcomp+7*\xgapd+\w+\xdec,-2.5*\w-2*\ys);            
    \draw [->] (3*\xcomp+5*\xgapd+\w+\xdec,-3.5*\w-3*\ys) -- (3*\xcomp+7*\xgapd+\w+\xdec,-3.5*\w-3*\ys);
    \draw [->] (3*\xcomp+5*\xgapd+\w+\xdec,-0.5*\w-\s) -- (3*\xcomp+7*\xgapd+\w+\xdec,-0.5*\w-\s); 
    \draw [->] (3*\xcomp+5*\xgapd+\w+\xdec,-1.5*\w-\ys-\s) -- (3*\xcomp+7*\xgapd+\w+\xdec,-1.5*\w-\ys-\s);
    
%---------- output inverse PA (Bob)
     \draw [->] (3*\xcomp+7*\xgapd+\w+\xdec+\xcomp,-2*\w-1.5*\ys-0.125*\s) -- (3*\xcomp+7*\xgapd+\w+\xdec+2*\xcomp,-2*\w-1.5*\ys-0.125*\s);      
     \draw [->] (3*\xcomp+7*\xgapd+\w+\xdec+\xcomp,-2*\w-1.5*\ys-0.375*\s) -- (3*\xcomp+7*\xgapd+\w+\xdec+2*\xcomp,-2*\w-1.5*\ys-0.375*\s);      
     \draw [->] (3*\xcomp+7*\xgapd+\w+\xdec+\xcomp,-2*\w-1.5*\ys-0.625*\s) -- (3*\xcomp+7*\xgapd+\w+\xdec+1*\xcomp+\c,-2*\w-1.5*\ys-0.625*\s);  
      \draw [->] (3*\xcomp+7*\xgapd+\w+\xdec+\xcomp,-2*\w-1.5*\ys-0.875*\s) -- (3*\xcomp+7*\xgapd+\w+\xdec+1*\xcomp+\c,-2*\w-1.5*\ys-0.875*\s);

 % Labeling
      \node at (3*\xcomp+7*\xgapd+\w+\xdec+2*\xcomp+5*\la,-2*\w-1.5*\ys-0.25*\s) {$\hat U^{KM}[\mathcal{F}_J] \equiv S_B^J$};  
      \node at (3*\xcomp+7*\xgapd+\w+\xdec+2*\xcomp+1*\la,-2*\w-1.5*\ys-0.75*\s) {$\hat U^{KM}[\mathcal{F}_J^{\setC}]$};            
      \node at (\xcomp+\xgapd-10*\la,-2*\w-1.5*\ys-0.25*\s) {$S_A^J\equiv U^{KM}[\mathcal{F}_J]$};  
      \node at (\xcomp+\xgapd-6*\la,-2*\w-1.5*\ys-0.75*\s) {$U^{KM}[\mathcal{F}_J^\setC]$};  
      
      \node at (2*\xcomp+3*\xgapd-5.5*\la,-1*\w-0.5*\ys) {$V^{L}[\mathcal{E}_K^\setC]$}; 
      \node at (2*\xcomp+3*\xgapd-5.5*\la,-3*\w-2.5*\ys) {$V^{L}[\mathcal{E}_K]$};       
      \node at (2*\xcomp+3*\xgapd-5.0*\la,-1*\w-0.5*\ys-\s) {$V^{ML}_{L+1}[\mathcal{E}_K]$}; 
      \node at (2*\xcomp+3*\xgapd-5.0*\la,-3*\w-2.5*\ys-\s) {$V^{ML}_{L+1}[\mathcal{E}_K^\setC]$};       
            
      \node at (3*\xcomp+6*\xgapd+\w+\xdec,-1*\w-0.5*\ys-\s) {$\hat V^{ML}_{L+1}[\mathcal{E}_K]$};             
      \node at (3*\xcomp+6*\xgapd+\w+\xdec,-3*\w-2.5*\ys) {$\hat V^{L}[\mathcal{E}_K]$};

% dashed boxes
%--------------------------------
% source
      \draw [dashed,gray] (3*\xcomp+4*\xgapd-\xdash,\ydash) -- (3*\xcomp+4*\xgapd+\w+\xdash,\ydash);      
      \draw [dashed,gray] (3*\xcomp+4*\xgapd-\xdash,-4*\w-3*\ys-\s-\ydash) -- (3*\xcomp+4*\xgapd+\w+\xdash,-4*\w-3*\ys-\s-\ydash);     
      \draw [dashed,gray] (3*\xcomp+4*\xgapd-\xdash,\ydash) -- (3*\xcomp+4*\xgapd-\xdash,-4*\w-3*\ys-\s-\ydash);       
      \draw [dashed,gray] (3*\xcomp+4*\xgapd+\w+\xdash,-4*\w-3*\ys-\s-\ydash) -- (3*\xcomp+4*\xgapd+\w+\xdash,\ydash);   
      \node at (3*\xcomp+4*\xgapd+0.5*\w,\ydash+4*\la) {\color{gray}{Source}};        
      \node at (3*\xcomp+3.5*\xgapd,\ydash+1.5*\la) {$X^N$};  
      \node at (3*\xcomp+4.5*\xgapd+\w,\ydash+1.5*\la) {$(Y^N,Z^N)$};                                     

% Tau_A
      \draw [dashed,gray] (3*\xcomp+3*\xgapd+\xdash,\ydash) -- (3*\xcomp+3*\xgapd+\xdash,-4*\w-3*\ys-\s-\ydash);  
      \draw [dashed,gray] (3*\xcomp+3*\xgapd+\xdash,\ydash) -- (\xgapd+\xcomp-\xdash,\ydash);  
      \draw [dashed,gray] (\xgapd+\xcomp-\xdash,\ydash) -- (\xgapd+\xcomp-\xdash,-0.66*\yPA-1.32*\ybs-0.66*\s);  
      \draw [dashed,gray] (\xgapd+\xcomp-\xdash,-0.66*\yPA-1.32*\ybs-0.66*\s) -- (\xgapd+\xcomp-\xdash-10*\la,-0.66*\yPA-1.32*\ybs-0.66*\s);        
      \draw [dashed,gray] (\xgapd+\xcomp-\xdash-10*\la,-4*\w-3*\ys-\s-\ydash) -- (\xgapd+\xcomp-\xdash-10*\la,-0.66*\yPA-1.32*\ybs-0.66*\s);  
      \draw [dashed,gray] (\xgapd+\xcomp-\xdash-10*\la,-4*\w-3*\ys-\s-\ydash) -- (3*\xcomp+3*\xgapd+\xdash,-4*\w-3*\ys-\s-\ydash);     
      \node at (1*\xcomp+1*\xgapd,-4*\w-3*\ys-\s-\ydash-\la) {\color{gray}{$\tau_A$}};    
      
 % Tau_B
      \draw [dashed,gray] (3*\xcomp+5*\xgapd+\w-\xdash,\ydash) -- (3*\xcomp+5*\xgapd+\w-\xdash,-4*\w-3*\ys-\s-\ydash);       
      \draw [dashed,gray] (3*\xcomp+5*\xgapd+\w-\xdash,\ydash) -- (4*\xcomp+7*\xgapd+\w+\xdec+\xdash,\ydash);           
      \draw [dashed,gray] (4*\xcomp+7*\xgapd+\w+\xdec+\xdash,\ydash) -- (4*\xcomp+7*\xgapd+\w+\xdec+\xdash,-0.66*\yPA-1.32*\ybs-0.66*\s);          
      \draw [dashed,gray] (4*\xcomp+7*\xgapd+\w+\xdec+\xdash,-0.66*\yPA-1.32*\ybs-0.66*\s) -- (4*\xcomp+7*\xgapd+\w+\xdec+\xdash+10*\la,-0.66*\yPA-1.32*\ybs-0.66*\s);  
      \draw [dashed,gray] (4*\xcomp+7*\xgapd+\w+\xdec+\xdash+10*\la,-4*\w-3*\ys-\s-\ydash) -- (4*\xcomp+7*\xgapd+\w+\xdec+\xdash+10*\la,-0.66*\yPA-1.32*\ybs-0.66*\s);  
      \draw [dashed,gray] (4*\xcomp+7*\xgapd+\w+\xdec+\xdash+10*\la,-4*\w-3*\ys-\s-\ydash) -- (3*\xcomp+5*\xgapd+\w-\xdash,-4*\w-3*\ys-\s-\ydash);  
      \node at (4.5*\xcomp+7*\xgapd+0.5*\xdec,-4*\w-3*\ys-\s-\ydash-\la) {\color{gray}{$\tau_B$}};

 % additional one-way communication form Alice to Bob
 \draw [] (2*\xcomp+3*\xgapd+0.5*\xcomp-7*\la,-1*\la) -- (2*\xcomp+3*\xgapd+0.5*\xcomp-7*\la,+7*\la); 
 \draw [] (2*\xcomp+3*\xgapd+0.5*\xcomp-7*\la,+7*\la) -- (3*\xcomp+5*\xgapd+\w+0.5*\xdec,+7*\la); 
 \draw [->] (3*\xcomp+5*\xgapd+\w+0.5*\xdec,+7*\la) -- (3*\xcomp+5*\xgapd+\w+0.5*\xdec,0);
 \node at (1*\xcomp+1.5*\xgapd+0.25*\xcomp-3.5*\la+1.5*\xcomp+2.5*\xgapd+0.5*\w+0.25*\xdec,+8.2*\la){$C_1$};  
 
 \draw [] (2*\xcomp+3*\xgapd+0.5*\xcomp-7*\la,-4*\w-3*\ys-\s-2*\la) -- (3*\xcomp+5*\xgapd+\w+0.5*\xdec,-4*\w-3*\ys-\s-2*\la); 
 \draw [] (2*\xcomp+3*\xgapd+0.5*\xcomp-7*\la,-4*\w-3*\ys-\s+1*\la) -- (2*\xcomp+3*\xgapd+0.5*\xcomp-7*\la,-4*\w-3*\ys-\s-2*\la);  
 \draw [->] (3*\xcomp+5*\xgapd+\w+0.5*\xdec,-4*\w-3*\ys-\s-2*\la) -- (3*\xcomp+5*\xgapd+\w+0.5*\xdec,-\yblock-2*\ybs-\s);                     
 \node at (1*\xcomp+1.5*\xgapd+0.25*\xcomp-3.5*\la+1.5*\xcomp+2.5*\xgapd+0.5*\w+0.25*\xdec,-4*\w-3*\ys-\s-3.2*\la){$C_2$};                      
                  
\end{tikzpicture}

%% file: pcc.tex
\def \xcomp{1.0}
\def \ycomp{5}
\def \xblock{2.93}
\def \yblock{3}

\def \yPA{1.75}

\def \xgapd{1.5} %gap between decompressors
\def \ygapd{0.9} %gap between decompressors

\def \ybs{-0.3} %height block 

\def \ys{0.1} % in dashed x-channel

\def \xdec{0.9} % xvalue decompressor
\def \w{0.526} % small W-block

\def \xs{0.4} % arrows at the border

\def \s{2.8} % shift for second PA block

\def \c{0.4} % frozen input length

\def \a{0.1} % antenna

\def \la{0.2} % for labeling

\def \xdec{1}

\def \xdash{0.2} %dist. dashed boxes
\def \ydash{0.2} %dist. dashed boxes

\def \xstart{0.8}

\begin{tikzpicture}[scale=1,auto, node distance=1cm,>=latex']
	
    % PA
     \draw [draw] (\xgapd+\xcomp,-2.5*\w-2*\ys+0.5*\w+0.5*\ys) -- (2*\xcomp+\xgapd,-2.5*\w-2*\ys+0.5*\w+0.5*\ys); 
     \draw [draw] (\xgapd+\xcomp,-2.5*\w-2*\ys+0.5*\w+0.5*\ys) -- (\xgapd+\xcomp,-\yPA-2*\ybs-\s);
     \draw [draw] (\xgapd+\xcomp,-\yPA-2*\ybs-\s) -- (2*\xcomp+\xgapd,-\yPA-2*\ybs-\s);
     \draw [draw] (2*\xcomp+\xgapd,-2.5*\w-2*\ys+0.5*\w+0.5*\ys) -- (2*\xcomp+\xgapd,-\yPA-2*\ybs-\s);     
    \node at (1.5*\xcomp+\xgapd,-1.32*\w-0.99*\ys-0.33*\yPA-0.66*\ybs-0.33*\s) {outer};
    \node at (1.5*\xcomp+\xgapd,-0.66*\w-0.495*\ys-0.66*\yPA-1.32*\ybs-0.66*\s) {$\enc$};
 %   \node at (1.5*\xcomp+\xgapd,-0.66*\yPA-1.32*\ybs-0.66*\s-2*\la) {$G_{KM}$};   

 % IR FIRST
  \draw [draw] (2*\xcomp+2*\xgapd,0) -- (2*\xcomp+2*\xgapd+\xcomp,0);  
  \draw [draw] (2*\xcomp+2*\xgapd,-\yblock-2*\ybs) -- (2*\xcomp+2*\xgapd+\xcomp,-\yblock-2*\ybs);
  \draw [draw] (2*\xcomp+2*\xgapd,0) -- (2*\xcomp+2*\xgapd,-\yblock-2*\ybs);  
 \draw [draw] (2*\xcomp+2*\xgapd+\xcomp,-\yblock-2*\ybs) -- (2*\xcomp+2*\xgapd+\xcomp,0);
     \node [] at (2*\xcomp+2*\xgapd+0.5*\xcomp,-0.33*\yblock-0.66*\ybs) {inner};
     \node [] at (2*\xcomp+2*\xgapd+0.5*\xcomp,-0.66*\yblock-1.32*\ybs) {$\enc$};
     
    \draw [draw] (3*\xcomp+3*\xgapd,0) -- (3*\xcomp+3*\xgapd+\w,0);  
    \draw [draw] (3*\xcomp+3*\xgapd,-\w) -- (3*\xcomp+3*\xgapd+\w,-\w);
    \draw [draw] (3*\xcomp+3*\xgapd,0) -- (3*\xcomp+3*\xgapd,-\w);
    \draw [draw] (3*\xcomp+3*\xgapd+\w,0) -- (3*\xcomp+3*\xgapd+\w,-\w);
   \node at (3*\xcomp+3*\xgapd+0.5*\w,-0.5*\w) {$\W$};
   
    \draw [draw] (3*\xcomp+3*\xgapd,-\ys-\w) -- (3*\xcomp+3*\xgapd+\w,-\ys-\w);  
    \draw [draw] (3*\xcomp+3*\xgapd,-\w-\ys-\w) -- (3*\xcomp+3*\xgapd+\w,-2*\w-\ys);
    \draw [draw] (3*\xcomp+3*\xgapd,-\ys-\w) -- (3*\xcomp+3*\xgapd,-2*\w-\ys);
    \draw [draw] (3*\xcomp+3*\xgapd+\w,-\ys-\w) -- (3*\xcomp+3*\xgapd+\w,-2*\w-\ys);
   \node at (3*\xcomp+3*\xgapd+0.5*\w,-1.5*\w-\ys) {$\W$};
   
    \draw [draw] (3*\xcomp+3*\xgapd,-2*\ys-2*\w) -- (3*\xcomp+3*\xgapd+\w,-2*\ys-2*\w);  
    \draw [draw] (3*\xcomp+3*\xgapd,-3*\w-2*\ys) -- (3*\xcomp+3*\xgapd+\w,-3*\w-2*\ys);
    \draw [draw] (3*\xcomp+3*\xgapd,-2*\ys-2*\w) -- (3*\xcomp+3*\xgapd,-3*\w-2*\ys);
    \draw [draw] (3*\xcomp+3*\xgapd+\w,-2*\ys-2*\w) -- (3*\xcomp+3*\xgapd+\w,-3*\w-2*\ys);
   \node at (3*\xcomp+3*\xgapd+0.5*\w,-2.5*\w-2*\ys) {$\W$};
   
    \draw [draw] (3*\xcomp+3*\xgapd,-3*\ys-3*\w) -- (3*\xcomp+3*\xgapd+\w,-3*\ys-3*\w);  
    \draw [draw] (3*\xcomp+3*\xgapd,-4*\w-3*\ys) -- (3*\xcomp+3*\xgapd+\w,-4*\w-3*\ys);
    \draw [draw] (3*\xcomp+3*\xgapd,-3*\ys-3*\w) -- (3*\xcomp+3*\xgapd,-4*\w-3*\ys);
    \draw [draw] (3*\xcomp+3*\xgapd+\w,-3*\ys-3*\w) -- (3*\xcomp+3*\xgapd+\w,-4*\w-3*\ys);
   \node at (3*\xcomp+3*\xgapd+0.5*\w,-3.5*\w-3*\ys) {$\W$};

 % IR SECOND
  \draw [draw] (2*\xcomp+2*\xgapd,-\s) -- (2*\xcomp+2*\xgapd+\xcomp,-\s);  
  \draw [draw] (2*\xcomp+2*\xgapd,-\yblock-2*\ybs-\s) -- (2*\xcomp+2*\xgapd+\xcomp,-\yblock-2*\ybs-\s);
  \draw [draw] (2*\xcomp+2*\xgapd,-\s) -- (2*\xcomp+2*\xgapd,-\yblock-2*\ybs-\s);  
 \draw [draw] (2*\xcomp+2*\xgapd+\xcomp,-\yblock-2*\ybs-\s) -- (2*\xcomp+2*\xgapd+\xcomp,-\s);
     \node [] at (2*\xcomp+2*\xgapd+0.5*\xcomp,-0.33*\yblock-0.66*\ybs-\s) {inner};
     \node [] at (2*\xcomp+2*\xgapd+0.5*\xcomp,-0.66*\yblock-1.32*\ybs-\s) {$\enc$};
          
    \draw [draw] (3*\xcomp+3*\xgapd,-\s) -- (3*\xcomp+3*\xgapd+\w,-\s);  
    \draw [draw] (3*\xcomp+3*\xgapd,-\w-\s) -- (3*\xcomp+3*\xgapd+\w,-\w-\s);
    \draw [draw] (3*\xcomp+3*\xgapd,-\s) -- (3*\xcomp+3*\xgapd,-\w-\s);
    \draw [draw] (3*\xcomp+3*\xgapd+\w,-\s) -- (3*\xcomp+3*\xgapd+\w,-\w-\s);
   \node at (3*\xcomp+3*\xgapd+0.5*\w,-0.5*\w-\s) {$\W$};
   
    \draw [draw] (3*\xcomp+3*\xgapd,-\ys-\w-\s) -- (3*\xcomp+3*\xgapd+\w,-\ys-\w-\s);  
    \draw [draw] (3*\xcomp+3*\xgapd,-\w-\ys-\w-\s) -- (3*\xcomp+3*\xgapd+\w,-2*\w-\ys-\s);
    \draw [draw] (3*\xcomp+3*\xgapd,-\ys-\w-\s) -- (3*\xcomp+3*\xgapd,-2*\w-\ys-\s);
    \draw [draw] (3*\xcomp+3*\xgapd+\w,-\ys-\w-\s) -- (3*\xcomp+3*\xgapd+\w,-2*\w-\ys-\s);
   \node at (3*\xcomp+3*\xgapd+0.5*\w,-1.5*\w-\ys-\s) {$\W$};
   
    \draw [draw] (3*\xcomp+3*\xgapd,-2*\ys-2*\w-\s) -- (3*\xcomp+3*\xgapd+\w,-2*\ys-2*\w-\s);  
    \draw [draw] (3*\xcomp+3*\xgapd,-3*\w-2*\ys-\s) -- (3*\xcomp+3*\xgapd+\w,-3*\w-2*\ys-\s);
    \draw [draw] (3*\xcomp+3*\xgapd,-2*\ys-2*\w-\s) -- (3*\xcomp+3*\xgapd,-3*\w-2*\ys-\s);
    \draw [draw] (3*\xcomp+3*\xgapd+\w,-2*\ys-2*\w-\s) -- (3*\xcomp+3*\xgapd+\w,-3*\w-2*\ys-\s);
   \node at (3*\xcomp+3*\xgapd+0.5*\w,-2.5*\w-2*\ys-\s) {$\W$};
   
    \draw [draw] (3*\xcomp+3*\xgapd,-3*\ys-3*\w-\s) -- (3*\xcomp+3*\xgapd+\w,-3*\ys-3*\w-\s);  
    \draw [draw] (3*\xcomp+3*\xgapd,-4*\w-3*\ys-\s) -- (3*\xcomp+3*\xgapd+\w,-4*\w-3*\ys-\s);
    \draw [draw] (3*\xcomp+3*\xgapd,-3*\ys-3*\w-\s) -- (3*\xcomp+3*\xgapd,-4*\w-3*\ys-\s);
    \draw [draw] (3*\xcomp+3*\xgapd+\w,-3*\ys-3*\w-\s) -- (3*\xcomp+3*\xgapd+\w,-4*\w-3*\ys-\s);
   \node at (3*\xcomp+3*\xgapd+0.5*\w,-3.5*\w-3*\ys-\s) {$\W$};

% decoder at Bob's side
    \draw [draw] (3*\xcomp+4*\xgapd+\w,0) -- (3*\xcomp+4*\xgapd+\w+\xdec,0); 
    \draw [draw] (3*\xcomp+4*\xgapd+\w,-\yblock-2*\ybs) -- (3*\xcomp+4*\xgapd+\w+\xdec,-\yblock-2*\ybs); 
    \draw [draw] (3*\xcomp+4*\xgapd+\w,0) -- (3*\xcomp+4*\xgapd+\w,-\yblock-2*\ybs);  
    \draw [draw] (3*\xcomp+4*\xgapd+\w+\xdec,-\yblock-2*\ybs) -- (3*\xcomp+4*\xgapd+\w+\xdec,0);  
    \node at (3*\xcomp+4*\xgapd+\w+0.5*\xdec,-0.33*\yblock-0.66*\ybs) {inner};         
    \node at (3*\xcomp+4*\xgapd+\w+0.5*\xdec,-0.66*\yblock-1.32*\ybs) {$\dec$};   
    
    \draw [draw] (3*\xcomp+4*\xgapd+\w,-\s) -- (3*\xcomp+4*\xgapd+\w+\xdec,-\s); 
    \draw [draw] (3*\xcomp+4*\xgapd+\w,-\yblock-2*\ybs-\s) -- (3*\xcomp+4*\xgapd+\w+\xdec,-\yblock-2*\ybs-\s); 
    \draw [draw] (3*\xcomp+4*\xgapd+\w,-\s) -- (3*\xcomp+4*\xgapd+\w,-\yblock-2*\ybs-\s);  
    \draw [draw] (3*\xcomp+4*\xgapd+\w+\xdec,-\yblock-2*\ybs-\s) -- (3*\xcomp+4*\xgapd+\w+\xdec,-\s);   
    \node at (3*\xcomp+4*\xgapd+\w+0.5*\xdec,-0.33*\yblock-0.66*\ybs-\s) {inner};         
    \node at (3*\xcomp+4*\xgapd+\w+0.5*\xdec,-0.66*\yblock-1.32*\ybs-\s) {$\dec$};       
% Bob's inverse PA block   
    \draw [draw] (3*\xcomp+5*\xgapd+\w+\xdec,-2.5*\w-2*\ys+0.5*\w+0.5*\ys) -- (3*\xcomp+5*\xgapd+\w+\xdec+\xcomp,-2.5*\w-2*\ys+0.5*\w+0.5*\ys);  
    \draw [draw] (3*\xcomp+5*\xgapd+\w+\xdec,-1.5*\w-\ys-\s-0.5*\w-0.5*\ys) -- (3*\xcomp+5*\xgapd+\w+\xdec+\xcomp,-1.5*\w-\ys-\s-0.5*\w-0.5*\ys);      
    \draw [draw] (3*\xcomp+5*\xgapd+\w+\xdec,-2.5*\w-2*\ys+0.5*\w+0.5*\ys) -- (3*\xcomp+5*\xgapd+\w+\xdec,-1.5*\w-\ys-\s-0.5*\w-0.5*\ys);   
    \draw [draw] (3*\xcomp+5*\xgapd+\w+\xdec+\xcomp,-1.5*\w-\ys-\s-0.5*\w-0.5*\ys) -- (3*\xcomp+5*\xgapd+\w+\xdec+\xcomp,-2.5*\w-2*\ys+0.5*\w+0.5*\ys);
    \node at (3*\xcomp+5*\xgapd+\w+\xdec+0.5*\xcomp,-1.32*\w-0.99*\ys-0.33*\yPA-0.66*\ybs-0.33*\s) {outer};    
    \node at (3*\xcomp+5*\xgapd+\w+\xdec+0.5*\xcomp,-0.66*\w-0.495*\ys-0.66*\yPA-1.32*\ybs-0.66*\s) {$\dec$};                   
% Arrows
            
    \draw [->] (2*\xcomp+\xgapd,-2.5*\w-2*\ys) -- (2*\xcomp+2*\xgapd,-2.5*\w-2*\ys);
    \draw [->] (2*\xcomp+\xgapd,-3.5*\w-3*\ys) -- (2*\xcomp+2*\xgapd,-3.5*\w-3*\ys);        
  
     % -----
         \draw [->] (\xcomp+\xgapd-\xstart,-2*\w-1.5*\ys-0.25*\s) -- (\xcomp+\xgapd,-2*\w-1.5*\ys-0.25*\s);  
         \draw [->] (\xcomp+\xgapd-\xstart,-2*\w-1.5*\ys-0.75*\s) -- (\xcomp+\xgapd,-2*\w-1.5*\ys-0.75*\s);    
     
  % ------      

    \draw [->] (2*\xcomp+\xgapd,-0.5*\w-\s) -- (2*\xcomp+2*\xgapd,-0.5*\w-\s); 
    \draw [->] (2*\xcomp+\xgapd,-1.5*\w-\ys-\s) -- (2*\xcomp+2*\xgapd,-1.5*\w-\ys-\s);     

  % ------   
    \draw [->] (3*\xcomp+2*\xgapd,-0.5*\w) -- (3*\xcomp+3*\xgapd,-0.5*\w); 
    \draw [->] (3*\xcomp+3*\xgapd +\w,-0.5*\w) -- (3*\xcomp+4*\xgapd+\w,-0.5*\w);
    
    \draw [->] (3*\xcomp+2*\xgapd,-1.5*\w-\ys) -- (3*\xcomp+3*\xgapd,-1.5*\w-\ys); 
    \draw [->] (3*\xcomp+3*\xgapd +\w,-1.5*\w-\ys) -- (3*\xcomp+4*\xgapd+\w,-1.5*\w-\ys);
        
    \draw [->] (3*\xcomp+2*\xgapd,-2.5*\w-2*\ys) -- (3*\xcomp+3*\xgapd,-2.5*\w-2*\ys); 
    \draw [->] (3*\xcomp+3*\xgapd+\w,-2.5*\w-2*\ys) -- (3*\xcomp+4*\xgapd+\w,-2.5*\w-2*\ys);
     
    \draw [->] (3*\xcomp+2*\xgapd,-3.5*\w-3*\ys) -- (3*\xcomp+3*\xgapd,-3.5*\w-3*\ys);             
    \draw [->] (3*\xcomp+3*\xgapd+\w,-3.5*\w-3*\ys) -- (3*\xcomp+4*\xgapd+\w,-3.5*\w-3*\ys);
  % ------   
    \draw [->] (3*\xcomp+2*\xgapd,-0.5*\w-\s) -- (3*\xcomp+3*\xgapd,-0.5*\w-\s); 
    \draw [->] (3*\xcomp+3*\xgapd+\w,-0.5*\w-\s) -- (3*\xcomp+4*\xgapd+\w,-0.5*\w-\s); 
    
    \draw [->] (3*\xcomp+2*\xgapd,-1.5*\w-\ys-\s) -- (3*\xcomp+3*\xgapd,-1.5*\w-\ys-\s); 
    \draw [->] (3*\xcomp+3*\xgapd +\w,-1.5*\w-\ys-\s) -- (3*\xcomp+4*\xgapd+\w,-1.5*\w-\ys-\s);
        
    \draw [->] (3*\xcomp+2*\xgapd,-2.5*\w-2*\ys-\s) -- (3*\xcomp+3*\xgapd,-2.5*\w-2*\ys-\s); 
    \draw [->] (3*\xcomp+3*\xgapd+\w,-2.5*\w-2*\ys-\s) -- (3*\xcomp+4*\xgapd+\w,-2.5*\w-2*\ys-\s);
     
    \draw [->] (3*\xcomp+2*\xgapd,-3.5*\w-3*\ys-\s) -- (3*\xcomp+3*\xgapd,-3.5*\w-3*\ys-\s);             
    \draw [->] (3*\xcomp+3*\xgapd+\w,-3.5*\w-3*\ys-\s) -- (3*\xcomp+4*\xgapd+\w,-3.5*\w-3*\ys-\s);
% ------------

    \draw [->] (3*\xcomp+4*\xgapd+\w+\xdec,-2.5*\w-2*\ys) -- (3*\xcomp+5*\xgapd+\w+\xdec,-2.5*\w-2*\ys);            
    \draw [->] (3*\xcomp+4*\xgapd+\w+\xdec,-3.5*\w-3*\ys) -- (3*\xcomp+5*\xgapd+\w+\xdec,-3.5*\w-3*\ys);
    \draw [->] (3*\xcomp+4*\xgapd+\w+\xdec,-0.5*\w-\s) -- (3*\xcomp+5*\xgapd+\w+\xdec,-0.5*\w-\s); 
    \draw [->] (3*\xcomp+4*\xgapd+\w+\xdec,-1.5*\w-\ys-\s) -- (3*\xcomp+5*\xgapd+\w+\xdec,-1.5*\w-\ys-\s);
    
%---------- output inverse PA (Bob)
     \draw [->] (3*\xcomp+5*\xgapd+\w+\xdec+\xcomp,-2*\w-1.5*\ys-0.25*\s) -- (3*\xcomp+5*\xgapd+\w+\xdec+1*\xcomp+\xstart,-2*\w-1.5*\ys-0.25*\s);      
     \draw [->] (3*\xcomp+5*\xgapd+\w+\xdec+\xcomp,-2*\w-1.5*\ys-0.75*\s) -- (3*\xcomp+5*\xgapd+\w+\xdec+1*\xcomp+\xstart,-2*\w-1.5*\ys-0.75*\s);

 % Labeling
      \node at (3*\xcomp+5*\xgapd+\w+\xdec+1*\xcomp+\xstart-0.5*\la,-2*\w-1.5*\ys-0.5*\s) {$\hat M^J$};          
      \node at (\xcomp+\xgapd-\xstart+0.5*\la,-2*\w-1.5*\ys-0.5*\s) {$M^J$};  
          
      \node at (3*\xcomp+2.5*\xgapd,1*\la) {$X^N$};  
      \node at (3*\xcomp+3.5*\xgapd+\w,1*\la) {$(Y^N\!,Z^N)$};                                     

      \node at (2*\xcomp+1.5*\xgapd,-2*\w-1.5*\ys+1*\la) {$T^{M}$};
      \node at (3*\xcomp+4.5*\xgapd+\xdec+\w,-2*\w-1.5*\ys+1*\la) {$\hat T^{M}$};  
      
      \node[gray] at (3*\xcomp+3*\xgapd+0.5*\w,4.2*\la) {$C_{1}$};
      \node[gray] at (3*\xcomp+3*\xgapd+0.5*\w,-\yblock-2*\ybs-\s-3.2*\la){$C_{2}$};

 % additional one-way communication form Alice to Bob
 \draw [dashed,gray] (2.5*\xcomp+2*\xgapd,0) -- (2.5*\xcomp+2*\xgapd,+3*\la); 
 \draw [dashed,gray] (2.5*\xcomp+2*\xgapd,3*\la) -- (3*\xcomp+4*\xgapd+\w+0.5*\xdec,3*\la); 
 \draw [->,dashed,gray] (3*\xcomp+4*\xgapd+\w+0.5*\xdec,3*\la) -- (3*\xcomp+4*\xgapd+\w+0.5*\xdec,0); 
 
 \draw [dashed,gray] (2.5*\xcomp+2*\xgapd,-\yblock-2*\ybs-\s) -- (2.5*\xcomp+2*\xgapd,-\yblock-2*\ybs-\s-2*\la); 
 \draw [dashed,gray] (2.5*\xcomp+2*\xgapd,-\yblock-2*\ybs-\s-2*\la) -- (3*\xcomp+4*\xgapd+\w+0.5*\xdec,-\yblock-2*\ybs-\s-2*\la);  
 \draw [->,dashed,gray] (3*\xcomp+4*\xgapd+\w+0.5*\xdec,-\yblock-2*\ybs-\s-2*\la) -- (3*\xcomp+4*\xgapd+\w+0.5*\xdec,-\yblock-2*\ybs-\s);

\end{tikzpicture}